\documentclass[11pt]{article}
 
\usepackage{amsthm, amsfonts, amsxtra, amssymb, amscd}

\usepackage{amsfonts,amsmath,amssymb,amscd,amsthm}
\usepackage[english]{babel} 
\usepackage[all]{xy}
\usepackage[backref, colorlinks, linktocpage, citecolor = blue, linkcolor = blue]{hyperref}

\usepackage{amsfonts,graphics, epic}
 \usepackage{graphicx}
\usepackage{multicol}
 \usepackage{tikz}
\usepackage{psfrag}
 \usepackage{young}
 \usepackage{ytableau}
\usepackage[enableskew]{youngtab}

\newtheorem*{lemma*}{Lemma}
\newtheorem{lemma}[subsection]{Lemma}
\newtheorem*{theorem*}{Theorem}
\newtheorem{theorem}[subsection]{Theorem}
\newtheorem*{proposition*}{Proposition}
\newtheorem{proposition}[subsection]{Proposition}

\newtheorem*{corollary*}{Corollary}

\newtheorem{corollary}[subsection]{Corollary}

\theoremstyle{definition}
\newtheorem*{definition*}{Definition}

\newtheorem*{example*}{Example}
\newtheorem{example}[subsection]{Example}

\theoremstyle{remark}
\newtheorem*{remark*}{Remark}
\newtheorem{remark}[subsection]{Remark}
 \newtheorem{definition}[subsection]{Definition}

\newcommand{\pd}[2]{\dfrac{\partial#1}{\partial#2}}

\renewcommand{\phi}{\varphi}

\newcommand{\be}{\begin{enumerate}}
\newcommand{\ee}{\end{enumerate}}

\title{ A note on the Formanek Weingarten function}
\author{Claudio Procesi}

\begin{document}\maketitle
\begin{abstract}
The aim of this   note is to compare work of Formanek \cite{formanek2} on  a certain construction of central polynomials with that of Collins \cite{Coll} on integration on unitary groups.

These two quite disjoint topics share the construction of the same function  on the symmetric group, which the second author calls {\em Weingarten function}. 

By joining these two approaches we succeed in giving a simplified and {\em very natural} presentation of both Formanek and Collins's Theory.
\end{abstract}
 
   \section{Schur Weyl duality}
   \subsection{Basic results} We need to recall some basic facts on the representation Theory of the symmetric and the linear group.
   
   Let $V$  be a vector space of finite dimension $d$ over a field $F$ which in this note can be taken as $\mathbb Q$  or $\mathbb C$.   On the tensor power $V^{\otimes k}$  act  both the symmetric group  $S_k$ and the linear group $GL(V)$, Formula \eqref{azio}, furthermore if $F=\mathbb C$ and $V$ is equipped with a Hilbert space structure one has an induced   Hilbert space structure on  $V^{\otimes k}$.  The unitary group $U(d)\subset GL(V)$ acts on  $V^{\otimes k}$ by unitary matrices.
   \begin{equation}\label{azio}
\sigma  \cdot u_{1}\otimes u_{2}\otimes \ldots \otimes u_{k}:=u_{\sigma^{-1}  (1)}\otimes u_{\sigma ^{-1}  (2)}\otimes \ldots \otimes u_{\sigma ^{-1}  (k)},$$$$ g  \cdot u_{1}\otimes u_{2}\otimes \ldots \otimes u_{k}:=gu_{1}\otimes gu_{2}\otimes \ldots \otimes gu_{k},\ \sigma\in S_k,\ g\in GL(V).
\end{equation}
The first step of Schur Weyl duality is the fact that the two operator algebras $ \Sigma_{k}(V) ,B_{k,d}$ generated respectively by $S_k$ and $GL(V)$  acting on $V^{\otimes k}$, are both semisimple and each the centralizer of the other.

In particular  the algebra $ \Sigma_{k}(V) \subset End(V^{\otimes k})=End(V)^{\otimes k}$ equals the subalgebra $\Sigma_{k}(V) =\left(End(V)^{\otimes k}\right)^{GL(V)}$ of  invariants  under the conjugation  action of the group $GL(V)\to End(V)^{\otimes k},\ g\mapsto g\otimes g\otimes\ldots \otimes g.$

From this, the double centralizer Theorem and work of Frobenius and Young one has that, 
  under the action of these two commuting groups,  the space $V^{\otimes k}$ decomposes into the direct sum          
   \begin{equation}\label{swd}
V^{\otimes k}=\oplus_{\lambda\vdash k,\ ht(\lambda)\leq d}M_\lambda\otimes S_\lambda(V)
\end{equation} over all partitions $\lambda$ of $k$  of height $\leq d$,  (the height $ht(\lambda)$ denotes the number of elements or {\em rows}, nonzero, of $ \lambda $).  \label{hei}

$M_\lambda$  is an irreducible representation of $S_k$ while $ S_\lambda(V)$, called a {\em Schur functor} is an irreducible polynomial representation of  $GL(V)$, which remains irreducible also when restricted to $U(d)$.  The partition with a single row $k$  corresponds   to the trivial representation of $S_k$  and to the symmetric power $S^k(V)$ of $V$.    The partition with a single column $k$  corresponds   to the sign representation of $S_k$  and to the exterior power $\bigwedge^k(V)$ of $V$       .\smallskip

The character theory of the two groups can be deduced from these representations.  We shall denote by $\chi_\lambda(\sigma)$ the character of the permutation $\sigma$  on $M_\lambda$. As for $S_\lambda(V)$ its character is expressed by a symmetric function  $S_\lambda(x_1,\ldots,x_d)$ restriction to the first $d$ variables of a stable symmetric function called {\em Schur function}.
Of this deep and beautiful Theory, see \cite{macdonald},  \cite{fulton},  \cite{fultonharris},  \cite{Weyl}, \cite{P7}, we shall use only two remarkable formulas, the {\em hook formula}  due to Frame,
Robinson and Thrall~\cite{sagan},  expressing the dimension   $\chi_\lambda(1)$ of $M_\lambda$  and the {\em hook-content formula} of  Stanley,  cf.  \cite[Corollary 7.21.4]{stanley})  expressing the dimension  $s_\lambda(d):= S_\lambda( 1,\ldots, 1)= S_\lambda( 1^d)$ of $S_\lambda(V)$.

We display partitions by Young diagrams, as in the figure below. 

By $\tilde\lambda$ we denote the dual partition obtained by exchanging rows and columns. The {\em boxes}, cf. \eqref{box},  of the diagram are indexed by pairs $(i,j)$ of coordinates.  \footnote{We use the {\em english notation}}  Given then one of the boxes $u$ we define its {\em hook number}  $h_u$ and its {\em content} $c_u$  as follows:

\begin{definition} \label{hook.numbers2}
Let $\lambda$ be a partition of $n$ and let $u=(i,j) \in \lambda$ be a
box in the corresponding Young diagram.  The {\em hook number}\index{hook number}
$h_u=h(i,j)$ and the {\em content} $c_u$ are defined as follows:
\begin{eqnarray} \label{eqn7}
h_u=h(i,j)= \lambda_i + \check \lambda_j -i-j +1,\quad c_u=c(i,j):=j-i.
\end{eqnarray}
\end{definition}
\begin{example}
Note that the box $u=(3,4)$ defines a hook in the diagram
$\lambda$, and $h_u$ equals the length (number of boxes) of this
hook:\bigskip

\bigskip

\ \  \qquad\ \ \ \  \qquad4
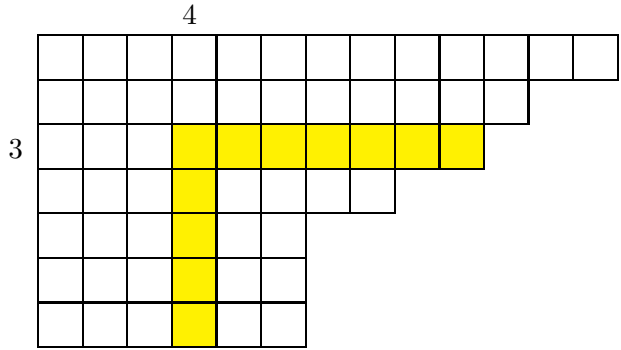
\smallskip

\ytableausetup{centertableaux}\label{YD}
\quad\ytableaushort{\none}
*{13, 11, 10, 8, 6, 6, 6}* 
[*(yellow)]{0,0,3+7,3+1,3+1,3+1,3+1}
\vskip-2.8cm 3

\vskip3cm
 
 In this figure, we have $\lambda=(13,11,10,8,6^3),\ ht(\lambda)=7$ with $u=(3,4)$.

Then $\check \lambda=(7^6, 4^2, 3^2, 2, 1^2)$ and $h_u=\lambda_3 +
\check \lambda_4 -3-4+1= 10 + 7 -6 =11$.

Here is another example:
In the following diagram of shape $\lambda=(8,3,2,1)$,  each hook
number $h_u$, respectively content $c_u$ is written inside its box in the diagram $\lambda$:
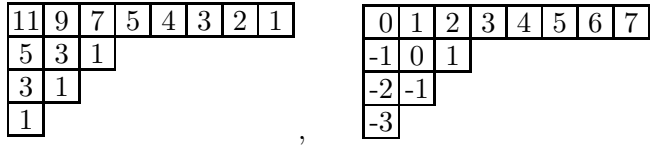
\bigskip

\vbox{\quad  \begin{Young}
11&9&7&5&4&3&2&1\cr
5&3&1\cr
3&1\cr
1\cr
\end{Young}\vskip-1.78cm  \quad\quad\quad\quad \quad \quad\quad\quad\quad\quad\quad\quad ,\quad \begin{Young}\label{conte}
0&1&2&3&4&5&6&7\cr
-1&0&1\cr
-2&-1\cr
-3\cr
\end{Young}}
\end{example}

\begin{theorem}[The hook and hook--content formulas] \label{thm2} Let $\lambda\vdash k$ be a\index{hook formula}
  partition of $k$ and $\chi_\lambda(1) $ and $ s_{\lambda }(d)$ be the dimension of the corresponding irreducible representation $M_\lambda$ of  $S_k$   and $S_\lambda(V)$ of $GL(V),\ \dim(V)=d$. Then
  \begin{equation}\label{sth}
s_{\lambda }(d)=\prod_{u\in \lambda}\frac{d+c_u}{h_u},\quad \chi_\lambda (1)=\frac{k!}{\prod_{u\in \lambda}h_u}.
\end{equation}
 
\end{theorem}
The remarkable Formula of Stanley, Theorem 15.3 of 
\cite{Stan0},                     exhibits 
$  s_{\lambda }(d)$ as a polynomial of degree $k=|\lambda|$  in $d$  with zeroes the integers $-c_u$  and leading coefficient  $\prod_{u\in \lambda}h_u^{-1}$, see \S \ref{shcf} for a proof.

\subsubsection{Matrix invariants}
The dual of the algebra  $End(V)^{\otimes k}$ can be identified, in a $GL(V)$ equivariant way, to   $End(V)^{\otimes k}$ by
the pairing formula:
$$\langle A_1\otimes A_2\dots\otimes A_k\mid B_1\otimes B_2\dots\otimes B_k\rangle:=tr(A_1\otimes A_2\dots\otimes A_k\circ
 B_1\otimes B_2\dots\otimes B_k)$$$$=tr(A_1B_1\otimes A_2B_2\dots\otimes A_kB_k)=\prod_{i=1}^k tr(A_iB_i).$$

Under this isomorphism the multilinear invariants of matrices are identified with the $GL(V)$  invariants
of  $End(V)^{\otimes m}$ which in turn are spanned by the elements of the symmetric group, hence by the elements of Formula \eqref{carsv}.
These are  explicited by Formula \eqref{phis} as in Kostant        \cite{kostant}.\begin{proposition}\label{lmuin} The space $\mathcal T_d(k)$ of  multilinear invariants of $k$, $d\times d$ matrices  is identified with  $End_{GL(V)}(V ^{\otimes k})$  and it is linearly spanned
by the functions:
\begin{equation}\label{carsv}
T_\sigma     (X_1,X_2,\dots,X_d):=tr(\sigma^{-1}\circ X_1\otimes X_2\otimes\dots\otimes   X_d),\ \sigma\in S_k.
\end{equation}
If $\sigma=(i_1i_2\dots i_h)\dots (j_1j_2\dots j_\ell)(s_1s_2\dots s_t)$ is the cycle
decomposition of $\sigma$ then  we have that
$T_\sigma     (X_1,X_2,\dots,X_d)$ equals \begin{equation}
\label{phis}=tr(X_{i_1}X_{i_2}\dots X_{i_h})\dots tr(X_{j_1}X_{j_2}\dots X_{j_\ell})
tr(X_{s_1}X_{s_2}\dots X_{s_t}).
\end{equation}
\end{proposition} \begin{proof}  
 Since the identity of Formula \eqref{phis} is multilinear it is enough to prove it on   the
decomposable tensors of $End(V)=V\otimes V^*$ which are the endomorphisms of rank 1,  $u\otimes \phi: v\mapsto \langle\phi\,|\,v\rangle u$.

So given $X_i:=u_i\otimes \phi_i$ and an element $\sigma\in S_k$ in the symmetric group we have
$$\sigma^{-1}\circ u_1\otimes \phi_1\otimes u_2\otimes \phi_2\otimes\ldots\otimes u_k\otimes \phi_k(v_1\otimes v_2\otimes \ldots\otimes v_k)$$
$$\stackrel{\eqref{azio}}= \prod_{i=1}^k\langle\phi_i\,|\,v_i\rangle u_{\sigma  (1)}\otimes u_{\sigma  (2)}\otimes \ldots \otimes u_{\sigma  (k)}
$$
$$u_1\otimes \phi_1\otimes u_2\otimes \phi_2\otimes\ldots\otimes u_k\otimes \phi_m\circ \sigma^{-1} (v_1\otimes v_2\otimes \ldots\otimes v_k)$$
$$= \prod_{i=1}^m\langle\phi_i\,|\,v_{\sigma(i)}\rangle u_{1}\otimes u_{2}\otimes \ldots \otimes u_{k}= \prod_{i=1}^k\langle\phi_{\sigma^{-1}(i)}\,|\,v_{i}\rangle u_{1}\otimes u_{2}\otimes \ldots \otimes u_{k}
$$$$\implies \sigma^{-1}\circ  u_1\otimes \phi_1\otimes u_2\otimes \phi_2\otimes\ldots\otimes u_m\otimes \phi_k= u_{\sigma  (1)}\otimes \phi_1\otimes u_{\sigma  (2)}\otimes \phi_2\otimes \ldots \otimes u_{\sigma  (k)}\otimes \phi_k
  $$\begin{equation}\label{formuu0}\implies  u_1\otimes \phi_1\otimes u_2\otimes \phi_2\otimes\ldots\otimes u_k\otimes \phi_k\circ \sigma=  u_1\otimes \phi_{\sigma(1)}\otimes u_2           \otimes \phi_{\sigma(2)}\otimes\ldots\otimes u_k\otimes \phi_{\sigma(k)}.
\end{equation} So we need to understand in matrix formulas the invariants \begin{equation}\label{formuu}
tr(\sigma^{-1} u_1\otimes \phi_1\otimes u_2\otimes \phi_2\otimes\ldots\otimes u_k\otimes \phi_k)=\prod_{i=1}^k\langle\phi_i\,|\,u_{\sigma(i)}\rangle  .
\end{equation} We need to use the rules
$$u\otimes\phi\circ v\otimes \psi=u\otimes\langle\phi\,|\,v\rangle\psi,\quad tr(u\otimes \phi)=\langle\phi\,|\,u\rangle $$ from which the formula easily follows by induction.\end{proof}
\begin{remark}\label{latt} We can extend the Formula   \eqref{carsv} to  the group algebra 
\begin{equation}\label{carsv1}
t(\sum_{\tau\in S_d }a_{ \tau}  \tau)(X_1,\ldots,X_d )
:=\sum_{\tau\in S_d }a_{\tau}   
T_\tau    (X_1,X_2,\dots,X_d).
\end{equation}\end{remark}
\subsection{The symmetric group}
The algebra of the symmetric group $S_k$ decomposes into the direct sum 
$$ F[S_k]=\oplus_{\lambda \vdash k} End(M_\lambda)$$of the matrix algebras associated to the irreducible representations $M_\lambda$ of partitions $\lambda\vdash k$.  Denote by $\chi_\lambda$   the corresponding character of $S_k$ and  by       $e_{\lambda}\in End(M_\lambda)\subset F[S_k]$ the corresponding central unit.  These elements form a basis of orthogonal idempotents of the center  of $ F[S_k]$. 

For a finite group $G$  let $e_i $ be the central idempotent of an irreducible representation with character $\chi_i $. One has the Formula:
\begin{equation}\label{unigr}
I)\quad e_{i}=\frac{\chi_i(1)}{|G|} \sum_{g\in G}\bar  \chi_i(g)g,\quad II)\quad \chi_i(e_j)=\begin{cases}
\chi_i(1)\quad \text{if}\  i=j\\0 \quad \quad \quad \text{if}\ i\neq j\\
\end{cases} .\end{equation} This is equivalent to the {\em orthogonality of characters}
\begin{equation}\label{ooc}
\frac 1{|G|} \sum_{g\in G})\bar \chi_i(g)\chi_j(g)=\delta^i_j.
\end{equation}
As for the algebra $ \Sigma_{k}(V) $,  it is isomorphic to  $ F[S_k]$  if and only if $d\geq k$. Otherwise it is a homomorphic image of $ F[S_k]$  with kernel the ideal generated by any antisymmetrizer in $d+1$ elements. This ideal is the direct sum of the $End(M_\lambda)$ with $ht(\lambda)>d$, where $ht(\lambda)$, the  {\em height} of $\lambda$, cf. page \pageref{hei} is also the length of its first column.   So that
\begin{equation}\label{akd}
 \Sigma_{k}(V) =\oplus_{\lambda \vdash k,\ ht(\lambda)\leq d} End(M_\lambda)
\end{equation}

\subsection{The function $W\!g(d,\mu)$}
We start with a computation of a character.
\begin{definition}\label{crho}
Given a permutation $\rho\in S_k$ we denote by $c(\rho)$ the number of cycles into which it decomposes, and $\pi(\rho)\vdash k$  the partition of $k$ given by the lengths   of these cycles.   Notice that $c(\rho)=ht(\pi(\rho))$.\label{piro}

Given a partition $\mu\vdash k$ we denote by
\begin{equation}\label{cmu}
(\mu):=\{\rho\mid\pi(\rho)=\mu\},\quad  C_{\mu}:=\sum_{\rho\mid \pi(\rho )=\mu}\rho=\sum_{\rho\in (\mu ) }\rho.
\end{equation}                     
\end{definition}
The sets $(\mu):=\{\rho\mid\pi(\rho)=\mu\}$ are the conjugacy classes of $S_k$  and, thinking of  $ F[S_k]$ as functions from $S_k$ to $F$ we have that $C_{\mu}$ is the characteristic function of the corresponding conjugacy class.
Of course the elements $C_\mu$ form a basis of the center of the group algebra   $ F[S_k]$.
\begin{proposition}\label{carro}1)\quad  For every pair of positive integers $k,d$ the function $P$ on  $S_k$ given by $P:\rho\mapsto d^{c(\rho)}$ is the character   of the permutation action of $S_k$ on $V^{\otimes k},$  $ \dim_{F}(V)=d$.

2)\quad  The symmetric bilinear form  on   $F[S_k]$  given by $\langle \sigma\mid\tau\rangle:=d^{c(\sigma\tau)}$  has as kernel the ideal generated by the antisymmetrizer on $d+1$ elements. In particular if $k\leq d$ it is non degenerate.

\end{proposition}
\begin{proof}
1)\quad  If $e_1,\ldots,e_d$ is a given basis  of $V$ we have the induced  basis   of  $V^{\otimes k}$, $e_{i_1}\otimes\ldots\otimes e_{i_k}$ which is permuted by the symmetric group. For a permutation representation  the  trace of an element $\sigma$ equals the number of the elements of the basis fixed by $\sigma$.

If $\sigma=(1,2,\ldots,k)$ is one cycle then $e_{i_1}\otimes\ldots\otimes e_{i_k}$  is fixed by $\sigma$  if and only if  ${i_1}={i_2}=\ldots= {i_k} $ are equal, so equal to some $e_j$ so $tr(\sigma)=d$. 

For a product of $a$ cycles  of lengths $b_1,b_2,\ldots b_a$  which up to conjugacy we may consider as
$$(1,2,\ldots,b_1)(b_1+1,b_1+2,\ldots,b_1+b_2)\ldots (k-b_a,\ldots,k)$$  we see  that
 $e_{i_1}\otimes\ldots\otimes e_{i_k}$  is fixed by $\sigma$  if and only if  it is of the form
 $$   e_{i_1}^{\otimes b_1}\otimes e_{i_2}^{\otimes b_2}\otimes\ldots\otimes e_{i_a}^{\otimes b_a} ,$$ giving $d^a$ choices  for the indices ${i_1},{i_2},\ldots,{i_a}$.

2)\quad  In fact this is the trace form of the image $ \Sigma_{k}(V) $ of  $F[S_k]$ in the operators on $V^{\otimes m},\ \dim V=d$. Since $ \Sigma_{k}(V) $ is semisimple its trace form is non degenerate.
\end{proof}

\begin{corollary}\label{dcro}
\begin{equation}\label{dcro1}
I)\quad P=\sum_{\lambda\vdash k,\ ht(\lambda)\leq d}s_{\lambda }(d)\chi_\lambda, \quad  II)\quad d^{c(\rho)}=\sum_{\lambda\vdash k,\ ht(\lambda)\leq d}s_{\lambda }(d)\chi_\lambda(\rho).
\end{equation}
\end{corollary}
\begin{proof}
This is immediate from Formula \eqref{swd}.
\end{proof}
We thus have, with $ht(\mu)$ the number of parts of $\mu$ (cf. page \pageref{piro}), that  \begin{equation}\label{lelep}
P:=\sum_{\rho\in S_k} d^{c(\rho)}\rho=\sum_{\mu\vdash k}d^{ht(\mu) }C_\mu
\end{equation}  is  an element of the center of the algebra   $ \Sigma_{k}(V) $ which we can thus write  \begin{equation}\label{lele}P=\sum_{\lambda\vdash k,\ ht(\lambda)\leq d}s_{\lambda }(d)\chi_\lambda=
\sum_{\rho\in S_k} d^{c(\rho)}\rho =\sum_{\lambda\vdash k,\ ht(\lambda)\leq d}r _{\lambda}(d)e_{\lambda}
\end{equation}  and we have:
\begin{proposition}\label{cdro}
\begin{equation}\label{cdro1}
r _{\lambda}(d)=\prod_{u\in\lambda}(d+c_u).
\end{equation}
\end{proposition}
\begin{proof}

  \label{weit}
 By Formula  \eqref{unigr}    we have:
 \begin{equation}\label{LeDu}
I)\quad e_{\lambda}=\frac{\chi_\lambda(1)}{k!} \sum_{\sigma\in S_k} \chi_\lambda(\sigma)\sigma,\quad II)\quad \chi_\lambda(e_\mu)=\begin{cases}
\chi_\lambda(1)\quad \text{if}\  \lambda=\mu\\0 \quad \quad \quad \text{if}\  \lambda\neq \mu\\
\end{cases} .\end{equation}

  One has thus, from Formulas \eqref{dcro1} I )  and    \eqref{LeDu} II) and denoting by  $(\chi_\lambda, P)$ the usual scalar product of characters:
$$ r _{\lambda}(d)=\frac{\sum_\rho d^{c(\rho)}\chi_\lambda( \rho)}{\chi_\lambda(1)}=\frac{ k!(  P,\chi_\lambda)}{\chi_\lambda(1)}=\frac{k!\,s_{\lambda }(d)}{ \chi_\lambda(1)} \stackrel{\eqref{sth}}=\prod_{u\in \lambda}(d+c_u).$$

\end{proof}

\begin{corollary}\label{invp}
The element $\sum_\rho d^{c(\rho)}\rho$ is invertible in $ \Sigma_{k}(V) $ with inverse
\begin{equation}\label{rmene}
( \sum_{\rho\in S_k}  d^{c(\rho)}\rho )^{-1}=\sum_{\lambda\vdash k,\ ht(\lambda)\leq d}(\prod_{u\in \lambda}(d+c_u))^{-1}e_{\lambda}.
\end{equation}
\end{corollary}

 As we shall see in \S \ref{coo}, it is interesting to study $(\sum_{\rho  \in S_k } d^{c(\rho)}\rho )^{-1}$ where $k$ is fixed and $d$ is a parameter. We can thus use formula  \eqref{rmene}  for $d\geq k$ and following Collins \cite{Coll} we write
\begin{equation}\label{wg}
(\sum_{\rho  \in S_k } d^{c(\rho)}\rho )^{-1}=\sum_{\rho  \in S_k } W\!g(d,\rho)\rho
:=W\!g(d,k) \end{equation} 
Since $W\!g(d,\rho)$  is a class function it depends only on the cycle partition $\mu=c(\rho)$ of $\rho$, so we may denote it by   $W\!g(d,\mu)$. We call the function  $W\!g(d,\rho)$ the {\em Formanek--Weingarten function}, since it was already introduced by Formanek in \cite{formanek2}.

From definition \eqref{cmu}   $C_\mu=\sum_{c(\rho)=\mu}\rho$ we can rewrite, $d\geq k$ 
\begin{equation}\label{wg0}C_\mu=\sum_{\rho\in S_k\mid c(\rho)=\mu}\!\!\!\!\!\!\!\rho,\qquad 
 W\!g(d,k)=(\sum_{\rho  \in S_k } d^{c(\rho)}\rho )^{-1}=\sum_{\mu \vdash k } W\!g(d,\mu)C_\mu
 .\end{equation}
Substituting $e_\lambda$ in  formula  \eqref{rmene}   with its expression of Formula  \eqref{LeDu}
$$ e_{\lambda}=\frac{\chi_\lambda(1)}{k!} \sum_{\sigma\in S_k} \chi_\lambda(\sigma)\sigma=\prod_{u\in\lambda}h_u^{-1} \sum_{\sigma\in S_k} \chi_\lambda(\sigma)\sigma$$  \begin{equation}\label{rmene1}
 W\!g(d,k):=\sum_{\rho  \in S_k } W\!g(d,\rho)\rho=\sum_{\lambda\vdash k}\prod_{u\in\lambda}\frac{ 1}{ h_u(d+c_u)} \sum_\tau \chi_\lambda(\tau)\tau\end{equation}
\begin{theorem}\label{wcom}
 \begin{equation}\label{rmene2} W\!g(d, \sigma)=  \sum_{\lambda\vdash k} \prod_{u\in\lambda}\frac{ 1}{ h_u(d+c_u)}\chi_\lambda(\sigma) = 
 \sum_{\lambda\vdash k} \frac{ \chi_\lambda(1)^2\chi_\lambda(\sigma)}{k!^2s_{\lambda }(d)}.\end{equation}
In particular $W\!g(d, \sigma)$ is a rational function of $d$ with poles at the integers $-k+1\leq i\leq k-1$  of order $p$ at $i$, $p(p+|i|)\leq k$. \end{theorem}
 \begin{proof}
We only need to prove the last estimate.  By symmetry we may assume that $i\geq 0$ then  the $p^{th}$ entry of $i$ is placed at the lower right corner of a rectangle of height $p$  and length $i+p$ (cf.  Figure at page \pageref{conte}). Hence if $\lambda\vdash k$, we have $i(p+i)\leq k$ and the claim.
\end{proof}
\subsubsection{A more explicit formula}
Formula \eqref{rmene2}, although explicit,   is a sum with alternating signs so that it is not easy to  estimate a given value or even to show that it is nonzero.

For $\sigma_0=(1,2,\ldots,k)$ a full cycle a better Formula is available. 
 First Formula \eqref{nfaaz} by Formanek  when $k=d$, and then   Collins Formula \eqref{nfaaz1} in general.

When $k=d$  we write  $W\!g(d, \sigma)=a_\sigma$ and then:

   \begin{equation}\label{nfaaz}
d!^2a_{\sigma_0}= (-1)^{d+1}\frac{d}{2d-1}\neq 0.
\end{equation}Collins  extends   Formula \eqref{nfaaz} to the case  $W\! g(d,\sigma_0)$ getting:
 \begin{equation}\label{nfaaz1}
W\! g(d,\sigma_0)=(-1)^{k-1}\mathtt C_{k-1}\prod_{-k+1\leq j\leq k-1}(d-j)^{-1}
\end{equation} with $\mathtt C_i:=\frac{(2i)!}{(i+1)!i!}=\frac{1}{i+1}\binom{2i}i$ the $i^{th}$ Catalan number. Which, since $$\mathtt C_{d-1}=\frac{(2d-2)!}{d!(d-1)!},\quad \prod_{-d+1\leq j\leq d-1}(d-j)=(2d-1)!$$  agrees, when $k=d$, with Formanek.

In order to prove  Formula \eqref{nfaaz1}  we need the fact  that
   $\chi_\lambda(\sigma_0)=0$ except when   $\lambda=(a,1^{k-a})$ is a {\em hook partition}, with the first row of some length  $a,\ 1\leq a\leq k$ and then the remaining $k-a$ rows of length 1.
   
   This is an easy consequence of the Murnaghan--Nakayama formula, see \cite{P7}.
   
    In this case we have    $\chi_\lambda(\sigma_0)=(-1)^ {k-a}$.
 We thus need to make explicit   the integers $s_{\lambda }(d),\chi_\lambda(1)$ for such a hook partition.

For $\lambda = (a,1^{k-a})$,   we get   that  the boxes are $$u=(1,j),\ j=1,\ldots, a,\ c_u=j-1,\ h_u=\begin{cases}
k\quad \text{if}\quad j=1\\  a-j+1\quad \text{if}\quad j\neq 1
\end{cases}$$ 
$$u=(i+1,1),\ i=1,\ldots, k-a,\quad   c_u= -i,\ h_u= k-a -i+1.\qquad$$
$$ \prod_uh_u=k\prod_{j=2}^a( a-j+1)\prod_{i=1}^{k-a}( k-a-i+1)=k(a-1)!({k-a} )!.$$
\begin{example}\label{box}

$a=8,\ k=11,\ (8,1^3)\vdash 11$ in coordinates\bigskip

\vbox{\quad  \begin{Young} 
1,1& 1\!,\!2&1\!,\!3&1\!,\!4&1\!,\!5&1\!,\!6&1\!,\!7&1\!,\!8\cr
2,1 \cr
3,1\cr
4,1\cr
\end{Young}  }
\bigskip

Hooks and content:
\bigskip

\vbox{\quad  \begin{Young}
11& 7&6&5&4&3&2&1\cr
3 \cr
2 \cr
1\cr
\end{Young}\vskip-1.78cm  \quad\quad\quad\quad \quad \quad\quad\quad\quad\quad\quad\quad ,\quad \begin{Young} 
0&1&2&3&4&5&6&7\cr
-1 \cr
-2 \cr
-3\cr
\end{Young}}\end{example} 
\bigskip

Thus we finally have, substituting in Formula \eqref{rmene2},  that 

  \begin{equation}\label{wg1}
W\!g(\sigma_0,d)=\sum_{a=1}^k (-1)^{k-a}\frac{1}{k(a-1)!({k-a} )! }\prod_{i=1-a}^{k-a }(d-i)^{-1}
\end{equation} 
 \begin{equation}\label{wg2}
= \sum_{a=1}^k (-1)^{k-a}\frac{\prod_{i= k-a+1}^{ k-1}(d-i) \prod_{i= -k+1}^{-a}(d-i) }{k(a-1)!({k-a} )!}\prod_{-k+1\leq j\leq k-1}(d-j)^{-1}.
\end{equation}

 One needs to show that 
$$  \sum_{a=1}^k (-1)^{ a}\frac{\prod_{i= k-a+1}^{ k-1}(d-i) \prod_{i= -k+1}^{-a}(d-i) }{k(a-1)!({k-a} )!}= \frac{\sum_{a=1}^k (-1)^{ a} \prod_{i= k-a+1}^{ k-1}i(d-i) \prod_{i=a}^{k-1}i(d+i) }{k!(k-1)!}$$
  \begin{equation}\label{ilcatt}
=P_k(d):= \frac1{k!}\sum_{b =0}^{k-1 } (-1)^{ b+1}  \binom { k-1 }{b } \prod_{i=  k -b }^{  k-1 } (d-i) \prod_{i=b+1}^{ k-1  } (d+i) =  (-1)^{k-1} \mathtt C_{k-1}.
\end{equation}
 By partial fraction decomposition we have that
 $$  \prod_{i=1-a}^{k-a }(d-i)^{-1} = \sum_{i=1-a}^{k-a }\frac{b_j}{d-j},$$$$b_0=\prod_{i=1-a,\ i\neq 0}^{k-a }( -i)^{-1}=[(-1)^{k-a}(a-1)!(k-a)!]^{-1}.  $$
 Therefore the partial fraction decomposition of  $ W\!g(\sigma_0,d)$, from Formula \eqref{wg1}, is
 $$ \sum_{a=1}^k\frac{1}{k[ (a-1)!(k-a)!]^2}\frac 1d+\sum_{-k+1\leq j\leq k-1,\ j\neq 0}\frac{c_j}{d-j}.$$
 On the other hand  the   partial fraction decomposition of  the product  of Formula \eqref{wg2},
 $$ \prod_{-k+1\leq j\leq k-1}(d-j)^{-1}=\frac {(-1)^{k-1}}{ (k-1)!^2 }\frac 1d+\sum_{-k+1\leq j\leq k-1,\ j\neq 0}\frac{e_j}{d-j}.$$
It follows that the polynomial $P_k(d)$ of Formula \eqref{ilcatt} is a constant $C$  with $$C\frac {(-1)^{k-1}}{ (k-1)!^2 }=\sum_{a=1}^k\frac{1}{k[ (a-1)!(k-a)!]^2}\implies C= (-1)^{k-1} \sum_{a=1}^k\frac{(k-1)!^2}{k[ (a-1)!(k-a)!]^2}.$$
 So finally we need to observe that 
 $$ \sum_{a=1}^k\frac{(k-1)!^2}{k[ (a-1)!(k-a)!]^2}=\frac{1} k \sum_{a=0}^{k-1}\binom{k-1}{ a }^2=\frac{1} k  \binom{2k-2}{ k-1 } =\mathtt C_{k-1}.$$
 In fact 
 $$\sum_{a=0}^{n}\binom{n}{ a }^2 =\binom {2n}n$$ as one can see simply noticing that a subset of $n$  elements  in $1,2,\ldots,2n$ distributes into $a$ numbers $\leq n$  and the remaining $n-a$  which are $>n$.
 
\qed  
 
\subsubsection{A Theorem of Collins, \cite{Coll} Theorem 2.2}For a partition $\mu\vdash k$ we have defined, in Formula \eqref{cmu}  
$C_{\mu}:=\sum_{\sigma\mid \pi(\sigma)=\mu}\sigma.$
Clearly we have for a sequence of partitions $ {\mu_1}, {\mu_2},\ldots, {\mu_i}$ 
\begin{equation}\label{prodcm}
C_{\mu_1} C_{\mu_2} \ldots C_{\mu_i}=\sum_{\mu\vdash k}A[\mu;\mu_1,\mu_2, \ldots,\mu_i]C_\mu
\end{equation} where $A[\mu;\mu_1,\mu_2, \ldots,\mu_i]\in\mathbb N$ counts the number of  times that a product of $i$ permutations  $  \sigma_1,\sigma_2, \ldots,\sigma_i$ of types  $  \mu_1,\mu_2, \ldots,\mu_i$ give a permutation $\sigma$ of type   $ \mu$.  These numbers are classically called {\em connection coefficients}.\begin{remark}\label{indep}
Notice that this number depends only on $\mu$ and not on $\sigma$.
\end{remark}  Set, for $i,h\in\mathbb N$:
 \begin{equation}\label{amui}
A[\mu,i,h]:=\sum_{ \substack{\mu_1,\mu_2, \ldots,\mu_i \mid \mu_j\neq 1^k\\\sum_{j=1}^i (k- ht(\mu_j))  = h  } }A[\mu;\mu_1,\mu_2, \ldots,\mu_i] \end{equation}
$$ A[\mu,h ]:=\sum_{i=1}^h(-1)^i A[\mu,i,h ].$$
  \begin{remark}\label{mintr}
For a permutation $\sigma\in S_k$ with $\pi(\sigma)=\mu$ we will write  \begin{equation}\label{mod}
|\sigma| =|\mu|:=k-ht(\mu).
\end{equation} This is the minimum number of transpositions with product $\sigma$ (see for this Proposition \ref{crf}).  

A minimal  product   of transpositions will also be called {\em reduced}.

We have $|\sigma\tau|\leq |\sigma|+|\tau|$, see Stanley \cite{Stan} p.446 for a poset interpretation.\end{remark}    
From Formula \eqref{rmene2} we know that each $W\!g(\sigma,d)$ is a rational function of $d$ with poles in      $0,\pm 1,\pm 2,\ldots ,\pm( k-1)$  of order $<k$, so we can expand it       in a power series in $d^{-1}$ converging for $d>k-1$  as in Formula \eqref{exp}:                        \begin{theorem}[\cite{Coll} Theorem 2.2]\label{invp1}
We have an expansion for $( \sum_{\rho\in S_k}  d^{c(\rho)}\rho )^{-1}$ as power series in $d^{-1}$:
\begin{equation}\label{exp}
 =d^{-k}(1+\sum_{\mu\vdash k}\left(\sum_{h=|\mu|}^\infty d^{-h} A[\mu, h] \right) C_\mu )
\end{equation}
 \end{theorem}
\begin{proof}
Recall that we denote by $|\mu|:=k-ht(\mu)$, \eqref{mod}.
$$
P=\sum_{\rho\in S_k}  d^{c(\rho)}\rho=d^k(1+\sum_{\mu\vdash k\mid \mu\neq 1^k}d^{-(k-ht(\mu))} C_\mu) =d^k(1+\sum_{\mu\vdash k\mid \mu\neq 1^k}d^{-|\mu|} C_\mu) $$ $$\text{so}\quad  P^{-1}=d^{-k}(1+\sum_{i=1}^\infty(-1)^i(\sum_{\mu\vdash k\mid \mu\neq 1^k}d^{-|\mu|} C_\mu)^i)$$
$$= d^{-k}(1+\sum_{i=1}^\infty(-1)^i(\sum_{\mu_1,\mu_2, \ldots,\mu_i \mid \mu_j\neq 1^k}d^{-\sum_{j=1}^i|\mu_j|} C_{\mu_1} C_{\mu_2} \ldots C_{\mu_i})$$
$$=d^{-k}(1+\sum_{\mu\vdash k}(\sum_{i=1}^\infty(-1)^ i\sum_{\mu_1,\mu_2, \ldots,\mu_i \mid \mu_j\neq 1^k}d^{-\sum_{j=1}^i|\mu_j|} A[\mu;\mu_1,\mu_2, \ldots,\mu_i] )C_\mu )$$
$$ =d^{-k}(1+\sum_{\mu\vdash k}\left(\sum_{h=|\mu|}^\infty d^{-h} A[\mu, h] \right)  C_\mu ) $$
since $\mu_1+\mu_2+ \ldots+\mu_i =\mu $ implies  $|\mu|\leq  \sum_{j=1}^i|\mu_j| $.\end{proof}

\begin{remark}\label{mintr1}

We want to see  now  that  the series  $\sum_{h=|\mu|}^\infty d^{-h} A[\mu, h]$  starts with  $h= |\mu|$, i.e. $ A[\mu, |\mu|]\neq 0$. Thus  we compute the leading coefficient $ A[\mu, |\mu|]$ which gives the asymptotic behaviour of  $W\!g(\sigma,d)$.
\end{remark}
Let us denote by
\begin{equation}\label{mmu}
C[\mu]:= A[\mu, |\mu|]\implies\lim_{d\to\infty}d^{k+|\sigma|}W\!g(\sigma,d)= C[\mu].
\end{equation}
From Formula \eqref{nfaaz} we have $C[(k)]=(-1)^{k-1}\mathtt C_{k-1}$ (Catalan number) and a further and more difficult Theorem of Collins states
\begin{theorem}\label{tcoll2}[\cite{Coll} Theorem 2.12 (ii)]\footnote{I have made a considerable effort trying to understand, and hence verify, the proof of this Theorem in  \cite{Coll}, to no avail.  To me it looks not correct.  Fortunately there is a  proof  in  \cite{Mn},   I will show presently   a simple natural proof.}
\begin{equation}\label{tcoll3}
C[(k)]=(-1)^{k-1}\mathtt C_{k-1},\qquad C[(a_1,a_2,\ldots,a_i)]=\prod_{j=1}^i C[(a_j)].
\end{equation}
\end{theorem}

Fixing   $\sigma\in S_k$ with $\pi(\sigma)=\mu$  we have that $A[\mu;\mu_1,\mu_2, \ldots,\mu_i]$ is also the number of sequences of permutations $\sigma_j,\ \pi(\sigma_j)=\mu_j$  with $\sigma=\sigma_1\sigma_2\ldots \sigma_i$.

 So we shall also use the notation, for $\pi(\sigma)=\mu$:   $$A[\sigma;\mu_1,\mu_2, \ldots,\mu_i]=A[\mu;\mu_1,\mu_2, \ldots,\mu_i],\quad C[\sigma]:= A[\sigma, |\sigma|].$$

Thus  \begin{equation}\label{amus}
C[\mu]=A[\mu, | \mu|]=\sum_{i=1}(-1)^ i\sum_{\substack{\mu_1,\mu_2, \ldots,\mu_i \mid \mu_j\neq 1^k\\\sum_{j=1}^i|\mu_j|=|\mu|}}  A[\mu;\mu_1,\mu_2, \ldots,\mu_i]
\end{equation} 
We call  a coefficient  $A[\mu;\mu_1,\mu_2, \ldots,\mu_i]
$  with $\mu_1,\mu_2, \ldots,\mu_i \mid \mu_j\neq 1^k ,$ and $ \sum_{j=1}^i|\mu_j|=|\mu|$ a {\em top coefficient}.

\subsubsection{ Top coefficients and a  degeneration of $\mathbb Q[ S_k]$}
The study of  $C[\mu]$ can be formulated in terms of a {\em  degeneration: $\mathbb Q[\tilde S_k]$} of the multiplication in the group algebra whose elements now denote by  $\tilde\sigma$. 

Define a new (still associative) multiplication   on $\mathbb Q[ S_k][q]$, $q$ a commuting variable by
\begin{equation}\label{disf0}
\mathbb Q[\tilde S_k]:=\oplus_{\sigma\in S_k}\mathbb Q\tilde\sigma,\quad \tilde\sigma_1\tilde\sigma_2:= q^{|\sigma_1|+| \sigma_2|-|\sigma_1 \sigma_2|} \widetilde{\sigma_1 \sigma_2}.\end{equation}
$$ ( \tilde\sigma_1\tilde\sigma_2) \tilde\sigma_3 = q^{|\sigma_1|+| \sigma_2|-|\sigma_1 \sigma_2|}  q^{|\sigma_1\sigma_2|+| \sigma_3|-|\sigma_1 \sigma_2\sigma_3|} \widetilde{\sigma_1 \sigma_2\sigma_3} $$$$=q^{|\sigma_1|+| \sigma_2| +| \sigma_3|-|\sigma_1 \sigma_2\sigma_3|} \widetilde{\sigma_1 \sigma_2\sigma_3}=\tilde\sigma_1(\tilde\sigma_2\tilde\sigma_3 ),\quad\text{associativity}.$$
When $q=1$ we recover the group algebra and when $q=0$  we have
\begin{equation}\label{disf}
\mathbb Q[\tilde S_k]:=\oplus_{\sigma\in S_k}\mathbb Q\tilde\sigma,\quad \tilde\sigma_1\tilde\sigma_2:=\begin{cases}
\widetilde{\sigma_1 \sigma_2}\ \text{if }|\sigma_1 \sigma_2|=|\sigma_1|+| \sigma_2|\\
0\quad \text{otherwise}
\end{cases}.
\end{equation}
Notice that, since $S_k$ is generated by transpositions  and $\tilde \tau^2=q^2$ for a transposition, we  have the algebra $\mathbb Q[ S_k][q^2]$.  

Further the product is compatible with the inclusions $S_k\subset S_{k+1}\subset\ldots$ so it defines an algebra on $\mathbb Q[\mathcal S] [q^2 ]$ where $\mathcal S=\cup_kS_k$.

Contrary to the semisimple algebra $  \mathbb Q[ S_k]$ the algebra  $  \mathbb Q[\tilde S_k]$ is a graded algebra, with  $  \mathbb Q[\tilde S_k]_h=     \oplus_{\sigma\in S_k\mid\,|\sigma|=h}\mathbb Q\tilde\sigma$ and has $$I:=\oplus_{\sigma\in S_k\mid\sigma\neq 1}\mathbb Q\tilde\sigma=\oplus_{h=1}^{k-1}\mathbb Q[\tilde S_k]_h $$ as  a nilpotent ideal, $I^k=0$, its nilpotent radical.
Observe that 
$$|\sigma_1 \sigma_2|=|\sigma_1|+| \sigma_2|\iff c(\sigma_1 \sigma_2)=c(\sigma_1)+c( \sigma_2)-k $$ so  if 
$c(\sigma_1)+c( \sigma_2)\leq k$  we know a priori that the product  $\tilde\sigma_1 \tilde\sigma_2=0$.    \smallskip                    

In this algebra the multiplication of two  elements $\tilde C_{\mu_1}, \tilde C_{\mu_2}$ associated to conjugacy classes 
as in   \eqref{cmu} involves only the top coefficients and is:

\begin{equation}\label{moto}
\tilde C_{\mu_1}  \tilde C_{\mu_2}=\sum_{ |\mu|=|\mu_1|+|\mu_2| }A[\mu;\mu_1,\mu_2]\tilde C_{\mu }.
\end{equation}

We then have
$$( \sum_{\rho\in S_k}  d^{c(\rho)}\tilde\rho )^{-1}=d^{-k}(1+\sum_{\mu\vdash k\mid \mu\neq 1^k}d^{-|\mu|} \tilde C_\mu)^{-1}=d^{-k}(1+\sum_{\mu\vdash k}  d^{-|\mu|}C[\mu  ]\tilde C_\mu )
$$\begin{equation}\label{laform}
=d^{-k}(1+\sum_{h=1}^{k-1}d^{- h}(\sum_{\mu\vdash k\mid |\mu|=h}  C[\mu  ]\tilde C_\mu )).\end{equation}
Notice that  if $h=k-1$  the only partition $\mu$ with $|\mu|=k-1$ is $\mu=(k)$  the partition of the full cycle. 

Hence in Formula \eqref{laform} the lowest term is $d^{-2k+1} C[(k)]\tilde C_{(k)}$.\smallskip

An example,which the reader can skip,  the connection coefficients for $S_4$,  in box the top ones (write the elements $C_\mu$  with lowercase):

      $$\!\!\!\!\!\!\!\!\!\!\begin{matrix}
&&c_{1, 1, 2}&c_{1, 3}&c_ {2, 2}&c_4\\\\ 
c_{1, 1, 2}&&  6c_ {1, 1, 1, 1}+ \boxed{3 c_{1, 3}+2c_{2, 2} }&  4c_ {1, 1, 2}+\boxed{4c_ {4} }&  c_{1, 1, 2}+\boxed{ 2c_{4}}& 3c_{1, 3}+4c_ {2, 2}
 \\
c_{1, 3}&&  4c_{1, 1, 2}+ \boxed{ 4c_ {4}} &8c_{1, 1, 1, 1}+4c_{1, 3}+8c_{2, 2}& 3c_{1, 3}& 4c_ {1, 1, 2} +4c_{4}\\
c_ {2, 2}&&c_{1, 1, 2}+ \boxed{2c_{4}}& 3c_ {1, 3}& 3c_ {1, 1, 1, 1}+2c_{2, 2}&2c_ {1, 1, 2}+ c_{4}\\
c_4&& 3c_{1, 3}+4c_{2, 2}& 4c_{1, 1, 2}+ 4c_ {4}&   2c_ {1, 1, 2}+  c_{4}&6c_ {1, 1, 
      1, 1}+3c_ {1, 3}+2c_ {2, 2}
\end{matrix} $$
Setting  $a=c_{1, 1, 2},\ b=c_{1, 3},\ c=c_ {2, 2},\ d=c_4$ compute Formula \eqref{laform}
 
 $$ a^2=3b+2c,\ ab= 4 d,\ ac= 2d$$
 $$P=1+T,\ T=x^{-1}a+x^{-2}(b+c)+x^{-3}d,\ (1+T)^{-1}=1-T+T^2-T^3 $$ 
$$ T^2=x^{-2}a^2+2x^{-3}a(b+c)=x^{-2}(3b+2c)+x^{-3}12d,\ T^3=x^{-3}a(3b+2c)=x^{-3}(12+4)d=x^{-3}16d  $$ 
$$-T+T^2-T^3= -x^{-1}a-x^{-2}(b+c)-x^{-3}d+x^{-2}(3b+2c)+x^{-3} 12d-x^{-3}16d 
$$
$$=  -x^{-1}a +x^{-2}(2 b+ c) -x^{-3} 5d$$
The conjugacy classes and their cardinality in $S_5$:
$$\left(
\begin{matrix} 
 1 ,c_{1,1,1,1,1} & 10 ,c_{1,1,1,2} &
 20 ,c_{1,1,3} &
 15 ,c_{1,2,2} &
 30 ,c_{1,4} &
 20 ,c_{2,3} &
 24 ,c_{5} &
\end{matrix}
\right)$$Here is a table of the top connection coefficients for $S_5$. The numbers to the right are the degrees  $|\mu|$:

$$ \begin{matrix}
a=c_{1,1,1,2} , 1&
 b=c_{1,1,3}, 2 &
c=c_{1,2,2}, 2 &
d=c_{1,4}, 3 &
 e=c_{2,3}, 3 &
 f=c_{5},4 &
\end{matrix} 
$$
$$ \begin{matrix}
&&c_{1,1,1,2}&c_{1,1,3}&c_{1,2,2}&c_{1,4}&c_{2,3}&c_5\\
c_{1,1,1,2}&&3c_{1,1,3}+2c_{1,2,2}&4c_{1,4}+c_{2.3}&2c_{1,4}+3c_{2.3}&5c_5&5c_5&0
\\c_{1,1,3}&&4c_{1,4}+c_{2.3}&5c_5&5c_5&0&0&0
\\c_{1,2,2}&&2c_{1,4}+3c_{2.3}&5c_5&5c_5&0&0&0
\\c_{1,4}&&5c_5&0&0&0&0&0
\\c_{2,3}&&5c_5&0&0&0&0&0
\\c_5&&0&0&0&0&0&0 
\end{matrix}$$
Compute Formula \eqref{laform}
 $$a^2=3b+2c,\ ab=4d+e,\ ac=2d+ 3e,\ ad=5f,\ ae=5f,\ b^2=5f,\ bc=5f,\ c^2=5f,\  $$
$$1+T,\ T=x^{-1}a+ x^{-2}(b+c)+x^{-3}(d+e)+x^{- 4}f$$
$$T^2=x^{-2}a^2+ x^{-4}(b+c)^2+2x^{-3}a(b+c)+2x^{- 4}a(d+e) $$
$$=x^{-2}(3b+2c) +2x^{-3} (6d+4e)+40x^{- 4}f    $$
$$T^3 =x^{-3}a(3b+2c) +2x^{-4} a(6d+4e) +x^{-4}(b+c)(3b+2c) $$
$$=x^{-3}( 12d+3e+4d+6e)+x^{-4}(100+15+10+15+10  )f$$
$$=x^{-3}( 16d+9e)+x^{-4}  150   f$$
$$T^4=x^{-4}a( 16d+9e)=x^{-4}(16\cdot 5+ 45 )f=x^{-4}125  f $$
$$125-150+40-1=14$$
 $\mathtt C_i$=Catalan(i): {1, 2, 5, 14, 42,\ldots}   Catalan(4)=14.
 
 $$-T+T^2-T^3+T^4=$$$$-(x^{-1}a+ x^{-2}(b+c)+x^{-3}(d+e) )+x^{-2}(3b+2c) +2x^{-3} (6d+4e) - x^{-3}( 16d+9e) +14f $$
 $$=-x^{-1}a- x^{-2}(b+c)-x^{-3}(d+e) +x^{-2}(3b+2c) +2x^{-3} (6d+4e) - x^{-3}( 16d+9e)$$
 $$=-x^{-1}a  +x^{-2}(3b+2c-b-c) + x^{-3} (12d+8e - 16d-9e-d-e)$$
 $$=-x^{-1}a  +x^{-2}(2b+ c) + x^{-3} (-5d  -2e) +14f.$$
 
\medskip
\subsubsection{Young subgroups }
Let $\Pi:=\{A_1,A_2,\ldots,A_j\},\ |A_i|=a_i$  be a decomposition of the set $[1,2,\ldots,k]$:
$$\text{i.e.}\quad A_1\cup A_2\cup \ldots\cup A_j=[1,2,\ldots,k] ,\ A_i\cap A_j=\emptyset,\ \forall i\neq j.$$
\begin{definition}\label{Ysb}\begin{enumerate}\item  The subgroup of $S_k$ fixing this decomposition  is the product  $\prod_{i=1}^jS_{A_i}= \prod_{i=1}^jS_{a_i}$ of the symmetric groups $S_{a_i}$. It is usually called a {\em Young subgroup} and will be denoted by $Y_{\Pi}$.  
 \item Given two decompositions  of $[1,2,\ldots,k]$, $\Pi_1:=\{A_1,A_2,\ldots,A_j\},$   and $  \Pi_2:=\{B_1,B_2,\ldots,B_h\}$ we say that $\Pi_1\leq\Pi_2$ if   each set $A_i$ is contained in one of the sets $B_d$. This is equivalent to the condition $Y_{\Pi_1}\subset Y_{\Pi_2}$.

\item 
 In particular, if  $\sigma\in S_k$ we  denote by $\Pi_\sigma$ the decomposition of $[1,2,\ldots,k]$ induced  by its cycles and denote $Y_\sigma:=Y_{\Pi_\sigma}$. 
 \end{enumerate}

\end{definition}  
\begin{remark}\label{lecly}
Observe that   $\tau\in Y_\Pi $ if and only if   $\Pi_\tau\leq\Pi $.   The conjugacy classes of $Y_\Pi$ are the  products of the conjugacy classes  in the blocks $A_i$. \smallskip

\end{remark}

 Then  we have for the group algebra and $\tau=(\tau_1, \tau_2,\ldots,\tau_j)\in Y_{\Pi}$:
 \begin{equation}\label{Ys}
\mathbb Q[Y_{\Pi}]=\otimes_{i=1}^j \mathbb Q[S_{a_i}]\subset \mathbb Q[S_{k}],\quad  (\tau_1, \tau_2,\ldots,\tau_j)=\tau_1\otimes\tau_2\otimes\ldots\otimes\tau_j .
\end{equation} We denote by   $\mathtt c_{\tau}$ the sum of the elements of the conjugacy class of $\tau$  in $Y_\Pi$ in order to distinguish it from $C_\tau$ the sum over the conjugacy class in $S_k$.
We have:
\begin{equation}\label{ccY}
\tau=   (\tau_1, \tau_2,\ldots,\tau_j)\in  Y_{\Pi},\ \mathtt c_{\tau}\stackrel{\eqref{cmu}}=C_{\tau_1}\otimes C_{\tau_2}\otimes\ldots\otimes C_{\tau_j }.
\end{equation}The first remark is:
 \begin{remark} \label{partb}
If $\tau=(\tau_1, \tau_2,\ldots,\tau_j)\in  Y_{\Pi}$ then for the number $c(\tau)$ of cycles of $\tau$ we have
\begin{equation*} 
c( \tau)=c(\tau_1)+c( \tau_2)+\cdots+c(\tau_j),\end{equation*}\begin{equation}\label{ciss1}
\implies |\tau|=\sum_ia_i-c(\tau)=\sum_i(a_i-c(\tau_i))=|\tau_1 |+|\tau_2|+\cdots+|\tau_j|.
\end{equation}
As a consequence if $\gamma=(\gamma_1, \gamma_2,\ldots,\gamma_j),\tau=(\tau_1, \tau_2,\ldots,\tau_j)\in  Y_{\Pi}$ we have
\begin{equation}\label{tasi}
|\gamma\tau|=|\gamma|+|\tau| \iff |\gamma_i\tau_i|=|\gamma_i|+|\tau_i|,\ \forall i.
\end{equation}
\end{remark}
If we then consider the associated discrete algebras, From Formulas \eqref{tasi} and  \eqref{Ys} we deduce  an analogous of  Formula \eqref{Ys} for the discrete algebras:
   \begin{equation}\label{Ys1}
\mathbb Q[\tilde Y_{\Pi}]=\otimes_{i=1}^j \mathbb Q[\tilde S_{a_i}]\subset \mathbb Q[\tilde S_{k}],\quad  \tau= (\tau_1, \tau_2,\ldots,\tau_j),  \ \tilde \tau =  \tilde \tau_1 \otimes \tilde \tau_2 \otimes\cdots\otimes \tilde \tau_j .
\end{equation}Formula \eqref{ciss1} tells us that $\mathbb Q[\tilde Y_{\Pi}]=\otimes_{i=1}^j \mathbb Q[\tilde S_{a_i}]$ as graded tensor product and the inclusion in $ \mathbb Q[\tilde S_{k}]$ preserves the degrees.
 
\subsubsection{A proof of Theorem \ref{tcoll2} }

In particular let $\sigma\in S_k$ and $\sigma=c_1c_2\ldots c_j$  its cycle decomposition. 

Let $A_i$ be the support of the cycle $c_i$ of $\sigma$ and $a_i$ its cardinality, so that $\Pi_\sigma=\{A_1,\ldots,A_j\} $ and $Y_\sigma=Y_{\Pi_\sigma}$.  
We have $\sigma\in Y_\sigma$ and its    conjugacy   class in $Y_\sigma$ is the product of the conjugacy classes  of the cycles $(a_i)\subset S_{a_i}$,  \eqref{cmu}. We denote, as before,  by   $\mathtt c_{\sigma}$ the sum of the elements of this conjugacy class.

We have now a very simple but crucial fact;
\begin{proposition}\label{crf}
\begin{enumerate}\item Let $(i,i_1,\ldots,i_a),\ (j,j_1,\ldots,j_b)$  be  two disjoint cycles, $a,b\geq 0$, and take the transposition  $(i,j)$ then:
\begin{equation}\label{spc}
(i,i_1,\ldots,i_a)    (j,j_1,\ldots,j_b)  (i,j)=(i,j_1,\ldots,j_b,j, i_1,\ldots,i_a )
\end{equation}
\begin{equation}\label{spc1}
 (i,j)   (i,i_1,\ldots,i_a)    (j,j_1,\ldots,j_b)=(j,j_1,\ldots,j_b, i, i_1,\ldots,i_a)
\end{equation}
\item Let $\sigma\in S_k$ and  $\tau=(i,j)$  a transposition.  Then $|\sigma\tau|=|\tau\sigma|=|\sigma|\pm1$  and  
$|\sigma\tau|=|\tau\sigma|=|\sigma|-1$ if and only if the two indices $i,j$  both belong to one of the sets of the partition of $\sigma$, i.e. $\tau=(i,j)\in Y_\sigma$. \end{enumerate}
\end{proposition}
\begin{proof}
1) is clear and 2) follows immediately from 1).  In fact either $i,j$ belong to the same cycle of $\sigma$ and then  in $\sigma\tau$ this cycle is split into two  and   $c(\sigma\tau)=c(\sigma)+1$  or $i,j$ belong to two different   cycles of $\sigma$ which are joined in $\sigma\tau$  and   $c(\sigma\tau)=c(\sigma)-1.$

Notice that, if $|\sigma\tau|=|\tau\sigma|=|\sigma|-1$,  $\Pi_{\sigma\tau}<\Pi_\sigma$ and is obtained from $\Pi_\sigma$  by replacing the support of the cycle in which $i,j$  appear with two subsets  support of the 2 cycles in which this splits. Similarly for  $\Pi_{\tau\sigma}$.
\end{proof}
From this we deduce the essential  result of this section:
\begin{corollary}\label{nys}
Let  $\sigma\in S_k$. Consider a decomposition $\sigma=  \sigma_1  \sigma_2 \ldots,\sigma_h,$ $ \sigma_i\in S_k,\ \sigma_i\neq 1,\forall i$  with
$|\sigma|=  |\sigma_1| + |\sigma_2|+ \ldots+|\sigma_h|$. Then for all $i$ we have $\sigma_i\in Y_{\Pi_\sigma}=Y_{ \sigma}$ (Definition \ref{Ysb}).\end{corollary}
\begin{proof}
By induction on $h$, if   $h=1$    there is nothing to prove.

If $\sigma_1=(i,j)$ is a transposition $|\sigma_1 |=1$,  then  the the claim follows by induction on $\sigma_1 \sigma=\bar\sigma=   \sigma_2 \ldots,\sigma_h,$ since $|\sigma_1 \sigma|=|\sigma|-1$ and Proposition  \ref{crf}.

If $|\sigma_1|>1$ we split  $\sigma_1=\tau \bar\sigma_1$  with $|\bar\sigma_1|= |\sigma_1|-1$ and $\tau$ a transposition and we are reduced to the previous case.
\end{proof}
We are now ready to prove  the Theorem of Collins,  Formula  \eqref{tcoll3}.

Let $\sigma\in S_k$ and $\sigma=c_1c_2\ldots c_j$  its cycle decomposition. Let $A_i$ be the support of the cycle $c_i$ and $a_i$ its cardinality, so that $\Pi_\sigma=\{A_1,\ldots,A_j\} $.   

By the previous Corollary \ref{nys}  and Remark \ref{indep} the contribution to $\sigma$ in the terms of Formula  \eqref{amui}   are  all in the subgroup $Y_\sigma$   so that finally $$\boxed{C[ \sigma] =C[\tilde\sigma]}\quad\text{with }C[\tilde\sigma]\  \text{computed in } \mathbb Q[\tilde Y_\sigma].$$   

In order to compute $C[\tilde\sigma]  $ we observe that the term  $d^{-k-|\tilde\sigma|} C[\tilde\sigma] \mathtt c_{\tilde \sigma}=d^{-k-|\sigma|} C[\tilde\sigma] \mathtt c_{\tilde \sigma}$  is the lowest  term in $d^{-1}$  in
\begin{equation}\label{mah}
( \sum_{\rho\in Y_\sigma}  d^{c(\rho)}\tilde\rho )^{-1}=\bigotimes_{i=1}^j( \sum_{\rho\in S_{a_i}}  d^{c(\rho)}\tilde\rho )^{-1}.
\end{equation}
From Formula  \eqref{laform} applied to the various full  cycles $c_i\in S_{a_i}$ we have that the lowest term in $( \sum_{\rho\in S_{a_i}}  d^{c(\rho)}\tilde\rho )^{-1}$ is $d^{-2a_i+1} C[(a_i)] C_{(a_i)}$ so that we have finally that the lowest term in Formula  \eqref{mah} is
\begin{equation}\label{lafine}
 d^{-k-|\sigma|} C[\tilde\sigma]\mathtt c_{\tilde \sigma}\stackrel{\eqref{ccY}}=\prod_{i=1}^j d^{-2a_i+1} C[(a_i)] C_{(a_1)}\otimes\ldots\otimes C_{(a_j)},\quad 
$$$$\implies C[ \sigma]=C[\tilde\sigma]=\prod_{i=1}^j   C[(a_i)] \stackrel{\eqref{nfaaz}}=\prod_{i=1}^j (-1)^{a_i-1}\mathtt C_{a_i-1}.\end{equation}
We have proved, Formula   \eqref{nfaaz}   that $(-1)^{a_i-1}C[(a_i)]$ is the Catalan number  $\mathtt C_{a_i-1}$ and the proof of Theorem \ref{tcoll2} is complete.\qed
 
\subsubsection{A table}

The case $k=d$ is of special interest, see \S \ref{sfor}. We  write   $ W\!g(d,\mu) =   a_\mu $ so that   $ \sum_{\mu\vdash d}   W\!g(d,\mu)c_\mu= \sum_\mu a_\mu c_\mu$  in Formula \eqref{wg0}. \bigskip

A   computation using Mathematica gives $d\leq 8$ the   list   $ d!^2\sum_{\mu\vdash d} a_\mu c_\mu$:

   $$ \frac 4{3} c_{1,1} - \frac 2{ 3}c_ {2} =\frac13(4c_{1,1}-2c_ {2})$$

   $$ \frac  {  21}{10}c_ {1^3}- \frac 9{ 10}c_ {1, 2} +\frac 3{5} c_{3} =\frac 1{10}(   21 c_ {1^3}-   9  c_ {1, 2} +6 c_{3}  ) $$

$$ \frac  { 134}{35}c_ {1^4} -\frac  { 48}{35}c_{1^2,2}  +\frac  { 29}{ 35}c_ {1, 3}   +\frac  {  22}{
  35}c_ {2^2} -\frac  { 4}{7}c_  {4}.$$
$$ \frac  1{35}(134 c_ {1^4} -  48 c_{1^2,2}  +\ 29 c_ {1, 3}   + 
  35 c_ {2^2} -20c_  {4}).$$
 The case $d=5:$ 
  
   $$  \frac   {145}{
 18}c_  {1^5} -\frac  { 299}{126}c_ {1^3, 2} + \frac  { 115}{
  126}c_ {1, 2^2}  + \frac  { 80}{63}c_ {1^2, 3} -\frac  { 101}{
   126}c_ {1, 4}  -\frac  { 37}{63}c_ {2, 3}+\frac  { 5}{9}c_ {5}    $$
 $$  \frac  1 {126}(         1015c_  {1^5} -299c_ {1^3, 2} + 160c_ {1^2, 3}+ 115c_ {1, 2^2}   -101c_ {1, 4}  -74c_ {2, 3}+70c_ {5})   $$
  The case $d=6:$  
  $$  \frac  {10508}{539}c_{1^6}  - \frac {2538}{539 }c_ {1^4, 2} +  \frac {1180}{
  539}c_{1^3, 3}+\frac  {2396}{ 1617}c_{1^2, 2^2}- \frac {668}{
   539}c_ {1^2, 4} -\frac {459}{ 539}c_ {1, 2, 3}  +\frac  {26}{
 33}c_ {1, 5}-\frac {338}{
   539}c_ {2^3} $$$$+\frac {922}{1617}c_{2, 4}  +\frac {300}{
  539}c_ {3, 3}-\frac {6}{11}c_ {6}$$
     
  $$  \frac 1{1617} (31524c_{1^6}  - 7614c_ {1^4, 2} + 3540c_{1^3, 3}+2396c_{1^2, 2^2}- 2004c_ {1^2, 4}$$$$ -1377c_ {1, 2, 3}  +1274c_ {1, 5}-1014c_ {2^3} +922c_{2, 4}  +900c_ {3, 3}-882c_ {6})$$
   
 The case $d=7:$
$$\frac {184849}{
  3432}c_{1^7}-\frac{12319}{1144}c_ {1^5, 2} +\frac{7385}{1716}c_{1^4, 3} +\frac {9401}{
  3432}c_{1^3, 2^2}-\frac{7369}{3432}c_ {1^3, 
   4} -\frac{196}{143}c_ {1^2,2, 3} +\frac {2107}{1716}c_{1^2, 5} $$$$ -\frac{1087}{1144}c_ {1, 2^3}+\frac {259}{
  312}c_{1, 2, 4} +\frac{1379}{1716}c_ {1, 3^2}-\frac{
   223}{286}c_ {1, 6}+\frac {1015}{1716}c_{2^2, 3}-\frac{961}{1716}c_ {2, 5} -\frac{85}{156}c_ {3, 4}+\frac {7}{
  13}c_{7} $$  
The biggest denominator $3432$  is also  a multiple of all denominators:
$$\frac 1{3432} (184849 c_{1^7}- 36957c_ {1^5, 2} +14770 c_{1^4, 3} +9401c_{1^3, 2^2}-7369c_ {1^3, 
   4} -4704c_ {1^2,2, 3} +4214c_{1^2, 5} $$$$ -3261c_ {1, 2^3}+2849c_{1, 2, 4} +2758c_ {1, 3^2}-2676c_ {1, 6}+2030c_{2^2, 3}-1922c_ {2, 5} -1870c_ {3, 4}+1848c_{7}  )$$  

 The case $d=8:$
$$  \frac{3245092}{19305}  c_{1^8} -\frac{546368}{19305}  c_{1^7, 2} +
 \frac{14434}{1485}  c_{1^5,
   3} +
 \frac{112828}{19305}  c_{1^4,
   2 , 2} -\frac{16336}{3861}  c_{1^4, 4}
   -\frac{4384}{1755}  c_{1^3, 2 , 3}
   +
 \frac{41332}{19305}  c_{1^3, 5} $$$$-\frac{10432}{6435}  c_{1^2, 2^3}
   +
 \frac{8608}{6435}  c_{1^2\, 2 , 4} +
 \frac{24718}{19305}  c_{1^2, 3^2} -\frac{2624}{2145}  c_{1^2, 6} +
 \frac{17122}{19305}  c_{1 , 2^2, 3} -\frac{1216}{1485}  c_{1 , 2 , 5} -\frac{1384}{1755}  c_{1 , 3 , 4} $$$$+
 \frac{151}{195}  c_{1 , 7} +
 \frac{124}{195}  c_{2^4} -\frac{11152}{19305}  c_{2^2, 4} -\frac{2176}{3861}  c_{2 , 3^2} +
 \frac{1186}{2145}  c_{2 , 6} +
 \frac{799}{1485}  c_{3 , 5} +
 \frac{796}{1485}  c_{4^2} -\frac{8}{15}  c_8$$

The biggest denominator $19305$  is also  a multiple of all denominators:
$$  \frac 1{19305}(3245092   c_{1^8} - 546368   c_{1^7, 2} +
187642 c_{1^5,
   3} +
112828  c_{1^4,
   2 , 2} -81680  c_{1^4, 4}
   -48224  c_{1^3, 2 , 3}
   +
 41332  c_{1^3, 5} $$$$-31296  c_{1^2, 2^3}
   +
 25824  c_{1^2\, 2 , 4} +
 24718  c_{1^2, 3^2} -23616  c_{1^2, 6} +
 17122  c_{1 , 2^2, 3} -15808 c_{1 , 2 , 5} -15224  c_{1 , 3 , 4} $$$$+
 14949   c_{1 , 7} +
12276 c_{2^4} -11152  c_{2^2, 4} -10880  c_{2 , 3^2} +
 10674  c_{2 , 6} +
10387  c_{3 , 5} +
 10348  c_{4^2} -10296 c_8).$$

  \smallskip

The reader will notice certain peculiar properties of these sequences.

First $W\!g(\sigma)$ is positive  (resp. negative) if $\sigma$ is  an even (resp. odd) permutation. This is a special case of a Theorem of Novak \cite{No},  Theorem \ref{Nov}.\smallskip
 
 {\bf Conjecture}  The absolute values  are strictly decreasing in the lexicographic order  of partitions written in increasing order.  The biggest denominator   is also  a multiple of all denominators.
 
 I verified this up to $d=14$.
 \subsection{ The results of Jucys Murphy and Novak}
 These conjectures  deserve further investigation, maybe the factorization of Jucys:
 \begin{equation}\label{JF}
\sum_{\rho\in S_k} d^{c(\rho|}\rho=d\prod_{i=2}^k(d+J_i),\quad J_i =(1,i)+(2,i)+\ldots+(i-1,i),\ i=2,\ldots,k 
\end{equation}  see \cite{Ju}  \cite{Mu} and the approach of Novak  \cite{No} can be used.\bigskip
 
 Let me give a quick exposition of these results:
 \begin{proposition}\label{JM}
The elements $J_i$ commute between each other.
\end{proposition}
\begin{proof}
This follows easily from the following fact, if $i<j<k$  then:
\begin{equation}\label{ijk}
(i,j)[(i,k)+(j,k)]=(i,j,k)+(j,i,k)=[(i,k)+(j,k)](i,j).
\end{equation}
\end{proof}
As for Formula  \eqref{JF} for $k=2$ it is clear and  then it follows by induction using the simple
\begin{lemma}\label{jmf} If $\sigma\in S_k\setminus S_{k-1}$  then 
$\sigma=\tau (i,k)$  where $\sigma(i)=k,\ i<k$  and $\tau\in S_{k-1},\ |\sigma|=|\tau|+1$ (from Proposition \ref{crf}  2.).

\end{lemma} \begin{proof}[Proof of Formula  \eqref{JF}]
Remark that, if $\rho\in S_{k-1}$,  the number of cycles of $\rho$,  thought of as element of $S_k$, is 1 more than  if   thought of as element of $S_{k-1} $ so,  by induction:
$$d\prod_{i=2}^k(d+J_i)=(\sum_{\rho\in S_{k-1}} d^{c(\rho)}\rho)(d+\sum_{i=1}^{k-1}(i,k))= (\sum_{\rho\in S_{k-1}\subset S_k} d^{c(\rho)}\rho)   + (\sum_{\rho\in S_k\setminus S_{k-1}  }d^{c(\rho)}\rho)$$
\begin{equation}\label{sodr}
 =\sum_{\rho\in   S_k} d^{c(\rho)}\rho=\sum_{j=1 }^k d^jC_j,\ C_j:=\sum_{\rho\in   S_k,\ c(\rho)=j}  \rho. 
\end{equation}
\end{proof}

Given this Novak observes that  in the Theory of symmetric functions, in the $k-1$  variables $x_2,\ldots,x_k$  we have
$$d\prod_{i=2}^k(d+x_i)=d^k+\sum_{i=1}^{k-1}d^{k-i}e_i(x_2,\ldots,x_k);\quad \prod_{i=2}^k(1-x_i)^{-1}=\sum_{j=0}^\infty h_j(x_2,\ldots,x_k)$$  where the $e_i(x_2,\ldots,x_k)$ are the elementary symmetric functions while  the $h_j(x_2,\ldots,x_k)$ are the {\em total symmetric functions}; that is  $h_j(x_2,\ldots,x_k)$ is the sum of all monomials  in  the variables $ x_2,\ldots,x_k $  of degree $j$. In particular   
$$c(\rho) =k-|\rho| \implies e_i(J_2 ,\ldots, J_k )=\sum_{\mu\vdash k\,\mid\,|\mu|= i} C_\mu.$$

Given this  one has  for $d\geq k$
$$(\sum_{\rho\in S_k} d^{c(\rho)}\rho)^{-1}=d^{-1}\prod_{i=2}^k(d+J_i)^{-1} =d^{-k}\sum_{j=0}^\infty h_j(-\frac{J_2}d,\ldots,-\frac{J_k}d) $$
\begin{equation}\label{finff}
=d^{-k}\sum_{j=0}^\infty \frac{(-1)^j}{d^j}h_j (J_2 ,\ldots, J_k )
\end{equation} a convergent   series for $d\geq k$.  This follows by remarking that setting
$$||\sum_\sigma a_\sigma  \sigma ||_\infty:=\max|a_\sigma|,\  ||AJ_i||_\infty\leq (k-1)||A ||_\infty $$
\begin{equation}\label{dise}
\implies || J_i^j||_\infty \leq (k-1)^j.
\end{equation}
This   series  in fact coincides with that given by Formula \eqref{exp}, but it is in many ways much better.

Observe that $h_j (J_2 ,\ldots, J_k )$  is a sum of permutations all with sign $(-1)^j$. Moreover since it is a   symmetric function                                                conjugate permutations appear with the same coefficient so it is a sum of $C_\mu$    for $\mu$  corresponding to   permutations of sign $(-1)^j$ with non negative integer coefficients. 
$$h_j (J_2 ,\ldots, J_k )=\sum_{\mu\vdash k\,\mid\, \epsilon(\mu)=(-1)^j} \alpha_{j,\mu} C_\mu,\ \alpha_{j,\mu}\in\mathbb N.$$ Split Formula \eqref{wg} as $$\sum_{\rho  \in S_k\mid \epsilon(\rho)=1 } W\!g(d,\rho)\rho
 =W\!g(d,k)_+;\quad \sum_{\rho  \in S_k\mid \epsilon(\rho)=-1 } W\!g(d,\rho)\rho
 =W\!g(d,k)_-$$ $$
\implies W\!g(d,k)_+= d^{-k}\sum_{j=0}^\infty \frac{1}{d^{2j}}h_{2j} (J_2 ,\ldots, J_k );$$\begin{equation}\label{Nov0}W\!g(d,k)_-= -d^{-k}\sum_{j=0}^\infty \frac{1}{d^{2j+1}}h_{2j+1} (J_2 ,\ldots, J_k ).
\end{equation}
\begin{theorem}\label{Nov}[Novak  \cite{No}]
$W\!g(d,\rho)>0$   if $\epsilon(\rho)=1$  and $W\!g(d,\rho)<0$   if $\epsilon(\rho)=-1$.
\end{theorem}
\begin{proof}
Let us give the argument for $\rho$ even  and $\pi(\rho)=\mu$. By Remark \ref{mintr1}: 
$$
 W\!g( d,\rho ) =W\!g( d,\mu) = d^{-k}\sum_{j=0}^\infty \frac{1}{d^{2j}}\alpha_{2j,\mu}=d^{-k-|\mu|}\sum_{j=0}^\infty \frac{1}{d^{2j}}\alpha_{2j+|\mu|,\mu}$$
the series $\sum_{j=0}^\infty \frac{1}{d^{2j}}\alpha_{2j+|\mu|,\mu}$ has the initial  term $\alpha_{ |\mu|,\mu}=C[ \mu ]$ and all positive terms so  $
 W\!g(\mu,d )\geq d^{-k-|\mu|}C[ \mu ]$.\end{proof}

\paragraph{Inequalities}
Let us describe some inequalities satisfied by  the function $W\!g(\sigma,d)$, let us write   for given $k,d$ by $W\!g(d,k)=\sum_\sigma W\!g(d,\sigma )\sigma =\Phi(1)^{-1}.$ From Formula \eqref{cdro1} since $ht(\lambda)\leq d$ we have 
$r _{\lambda}(d)=\prod_{u\in\lambda}(d+c_u)>0$. So $P$ and $P^{-1}=W\!g(d,k)$  are both positive symmetric operators.  We start with
\begin{proposition}\label{inq}
\begin{equation}\label{trss}
W\!g(\sigma,d)=tr(\sigma^{-1} W\!g(d,k)^2).
\end{equation}
\end{proposition}
\begin{proof}
$$\sum_\sigma tr(\sigma^{-1} W\!g(d,k)^2)\sigma= \Phi(W\!g(d,k)^2)=\Phi(1)\Phi(1)^{-2}=\Phi(1)^{-1}.$$
\end{proof}
 
Now in the space $V=\mathbb R^d$  consider the usual scalar product under which the basis $e_i$ is orthonormal. Remark that in the algebra of operators  $\Sigma_k(V)$  we have, for $\sigma\in S_k$  that the transpose of  $\sigma$ is $\sigma^{-1}$, by Formula \eqref{tras}.
\begin{equation}\label{tras}
(u_1\otimes\cdots\otimes u_k,\sigma\circ  v_1\otimes\cdots\otimes v_k)=\prod_{i=1}^k (u_i,v_{\sigma^{-1}(i)}=\prod_{i=1}^k (\sigma (u_i),v_{i}) .\end{equation}
Next  we have that $W\!g(d,k)$  and $W\!g(d,k)^2$ are  positive symmetric operators.

In the algebra $\Sigma_k(V)$, a sum of matrix algebras over $\mathbb R$, the nonnegative symmetric elements are of the form  $aa^t,\ a\in \Sigma_k(V)$ so that we have  
\begin{proposition}\label{inq1}
\begin{equation}\label{trss1}
 tr(aa^t W\!g(d,k)^2)\geq 0,\ \forall a\in \Sigma_k(V).\end{equation}
 This implies that, given any element $0\neq  \sum_{\sigma\in S_k}a_\sigma$ setting
 $$ \sum_{\sigma\in S_k}b_\sigma\sigma:=( \sum_{\gamma\in S_k}a_\gamma \gamma)( \sum_{\tau \in S_k}a_\tau\tau^{-1}),\ b_\sigma=\sum_{\gamma,\tau\mid\gamma\tau^{-1}=\sigma}a_\gamma a_\tau$$
 $$\implies  \sum_{\sigma\in S_k}b_\sigma W\!g(\sigma,d)> 0,\ $$
\end{proposition}
\begin{example}\label{inn}
$(1\pm \sigma)(1\pm\sigma^{-1})=2\pm(\sigma +\sigma^{-1})$ gives  
$$ W\!g(1,d)>W\!g(\sigma,d)>-W\!g(1,d),\ \forall \sigma\neq 1.$$

 \end{example}

\subsection{The algebra $(\bigwedge M_d^*)^G$}

Preliminary to the next step we need to recall  the theory of  antisymmetric  conjugation   invariant functions on $M_d$. This is a classical theory over a field of characteristic 0 which one may take as $\mathbb Q$. 

First, let $U$ be a vector space. A polynomial  $g(x_1,\ldots,x_{m})$  in $m$ variables   $ x_i\in U$  is {\em  antisymmetric} or {\em alternating} in the variables $X:=\{  x_1,\ldots,x_{m}\}$ if for all permutations $\sigma\in S_m$ we have
$$g(x_{\sigma(1)},\ldots,x_{\sigma(m)})=\epsilon_\sigma g(x_1,\ldots,x_{m}),\ \epsilon_\sigma\ \text{the sign of}\ \sigma. $$
A simple way of forming  an antisymmetric  polynomial  from a given one  $g(x_1,\ldots,x_{m})$ is the process of {\em alternation}\footnote{we avoid on purpose multiplying by $1/m!$}
\begin{equation}\label{alta}
Alt_Xg(x_1,\ldots,x_{m}):=\sum_{\sigma\in S_m} \epsilon_\sigma g(x_{\sigma(1)},\ldots,x_{\sigma(m)}).
\end{equation}  Recall that the exterior algebra $\bigwedge U^*$, with $U$ a vector space,   can be thought of as the space   of multilinear  alternating functions on $U$. Then exterior multiplication   as functions is given by the Formula:
 $$f(x_1,\ldots,x_h)\in\bigwedge^h U^* ;\quad g(x_1,\ldots,x_k)\in \bigwedge^k U^*,$$
 \begin{equation} 
\ f\wedge g(x_1,\ldots , x_{h+k})=\frac{1}{h!k!}\sum_{\sigma\in S_{h+k}}f(x_{\sigma(1)},\ldots,x_{\sigma(h)})  g(x_{\sigma(h+1)},\ldots,x_{\sigma(h+k)}) \end{equation} \begin{equation}\label{ep}
=\frac{1}{h!k!}Alt_{x_1,\ldots,x_{h+k}} f(x_1,\ldots,x_h) g(x_{ h+1 },\ldots,x_{ h+k })\in \bigwedge^{h+k} U^*.
\end{equation}
  It is well known that:
\begin{proposition}\label{ilmulti}
A  multilinear and antisymmetric  polynomial  $g(x_1,\ldots,x_{m})$ in $m$ variables   $ x_i\in \mathbb C^m$ is a multiple,    $a\det(x_1,\ldots,x_{m}) $, of the determinant.

In fact  if the polynomial has integer coefficients $a\in\mathbb Z$.
\end{proposition}  For a multilinear and antisymmetric  polynomial map $g(x_1,\ldots,x_{m})\in U$ to a vector space, each coordinate has the same property so $$g(x_1,\ldots,x_{m})= \det(x_1,\ldots,x_{m})a,\ a\in U.$$
 We apply this to $U=M_d$. Let us identify $M_d=\mathbb C^{d^2}$ using the canonical basis  of  elementary matrices $e_{i,j}$  ordered   lexicographically e.g.:
$$d=2,\quad e_{1,1},\  e_{1,2},\  e_{2,1},\  e_{2,2}. $$

Given $d^2$ matrices  $ Y_1,\ldots, Y_{d^2}\in M_d$  we may consider them as elements of $\mathbb C^{d^2}$ and then form the determinant $\det( Y_1,\ldots, Y_{d^2})$. 

\medskip

By  Proposition \ref{ilmulti} the 1 dimensional  space $\bigwedge^{d^2 } M_d^*$  has as generator the determinant $\det( Y_1,\ldots, Y_{d^2})$ which, since the conjugation action by $G:=GL(d,\mathbb Q)$ on $M_d$ is by transformations of determinant 1,  is thus an invariant  under the action by $G  $. 

The theory of $G$  invariant antisymmetric  multilinear  $G$  invariant functions on $M_d$ is well known and related to the cohomology of $G$.

The  antisymmetric  multilinear  $G$  invariant functions on $M_d$ form the algebra $(\bigwedge M_d^*)^G$. This is a subalgebra of the exterior algebra $ \bigwedge M_d^*$  and can be identified to the cohomology of the unitary group. As all such cohomology algebras it is a Hopf algebra and by Hopf's Theorem it is the exterior algebra generated by the  primitive elements.

The  primitive elements  of   $(\bigwedge M_d^* )^G$ are, see \cite{kostant}:
   \begin{equation}\label{prime}
T_{2i-1}=T_{2i-1}(Y_1,\ldots,Y_{2i-1}):=tr (St_{2i-1}(Y_1,\ldots,Y_{2i-1}))
\end{equation}$$St_{2i-1}(Y_1,\ldots,Y_{2i-1})= \sum_{\sigma\in S_{2i-1} }\epsilon_\sigma Y_{\sigma(1)} \ldots Y_{\sigma(2i-1)} $$ with $i=1,\ldots,d$.
In particular, since these elements generate an exterior algebra we have: 
\begin{remark}\label{extal}
A product of   elements $T_i$ is non zero if and only if the $T_i$ involved are all distinct, and then it depends on the order only up to a sign. 
\end{remark}The $2^n$ different products form a basis of  $(\bigwedge M_d^* )^G$.   The  non zero product of all these elements $T_{2i-1}(Y_1,\ldots,Y_{2i-1})$ is in dimension $d^2$.   We denote 
 \begin{equation}\label{ILTD}
\mathcal T_d(Y_1,Y_2,\ldots,Y_{d^2})=T_1\wedge T_3\wedge T_5\wedge \cdots \wedge T_{2d-1}.
\end{equation} 
\begin{proposition}\label{mmul}
A multilinear  antisymmetric function of  $Y_1,\ldots,Y_{d^2 }$   is a multiple of $T_1\wedge T_3\wedge T_5\wedge \cdots \wedge T_{2d-1}$.

\end{proposition} 
\begin{remark}\label{Iin} The function $\det( Y_1,\ldots, Y_{d^2})$  is an invariant of matrices so it must have an expression as in Formula \eqref{phis}.  In fact up to a computable integer constant \cite{formanek2} this equals  the exterior product of Formula \eqref{ILTD}.\end{remark}
The constant of the change of basis when we take as basis the matrix units can be computed up to a sign, see \cite{formanek2}:
\begin{equation}\label{costdis}\mathcal T_d
(Y)=\mathcal  C_d\det( Y_1,\ldots, Y_{d^2}),\quad 
\mathcal C_d:=\pm \frac{1!3!5!\cdots (2d-1)!}{1!2!\cdots (d-1)!}.
\end{equation}

  \section{Comparing Formanek, \cite{formanek2} and Collins \cite{Coll}}
Rather than following the historical  route we shall first discuss the paper of Collins, since this will allow us to introduce some notations  useful for the discussion of Formanek's results.

\subsection{The work of Collins\label{coo}}
In the  paper  \cite{Coll}, Collins   introduces  the Weingarten function in the following context. He is interested in computing integrals of the form
 \begin{equation}\label{intr}
\int_{U(d)}\prod_{\ell=1}^{k_1}u_{j_\ell,h_\ell}\prod_{m=1}^{k_2}\bar u_{i_m,p_m} du
\end{equation}
where $U(d)$  is the unitary group of $d\times d$ matrices and the elements  $u_{i,j}$  the  entries of a matrix $X\in U(d)$ while     $\bar u_{j,i}$  the  entries of  $X^{-1}=U^*=\bar U^t$. Here  $du$ is the normalized Haar measure. If one translates  by a scalar matrix $\alpha,\ |\alpha|=1$  then the integrand is multiplied by $\alpha^{k_1}\bar\alpha^{k_2}$, on the other hand Haar measure is invariant under multiplication so that this integral vanishes unless  we have $k_1=k_2$. In this case  the computation will be algebraic based on the following considerations.\smallskip

Let us first make some general remarks.  A finite dimensional  representation $R$  of a compact group $G$ (with the dual denoted by $R^*$), decomposes  into the direct sum of irreducible representations.  In particular  if $R^G$  denotes the subspace of $G$  invariant vectors there is a canonical $G$ equivariant projection $E:R\to R^G$. The projection $E$  can be written as integral 
\begin{equation}\label{inth}
E(v):=\int_G g\cdot v\,dg,\quad dg\quad \text{normalized Haar measure.}
\end{equation}  In turn the integral $ E( v)=\int_G g\cdot v\,dg$  is defined in dual coordinates by
\begin{equation}\label{info}
\langle\phi\mid E( v)\rangle=\langle\phi\mid \int_G g\cdot v\,dg\rangle:= \int_G \langle\phi\mid g\cdot v\rangle dg ,\ \forall \phi\in R^*.  
\end{equation}  The functions, of $g\in G$,   $ \langle\phi\mid g\cdot v\rangle,\  \phi\in R^*,\ v\in R$ are called {\em representative functions}; therefore an explicit formula for $E$ is equivalent to the knowledge of integration of  representative functions. In fact  usually the integral is computed by some algebraic method of computation of $E$.

Consider the space $V=\mathbb C^d$ with natural basis $e_i$ and dual basis $e^j$. 

We take  $R=End(V)$  with the conjugation action of  $GL(V)$ or of  its compact subgroup  $U(d)$  of unitary $d\times d$ matrices:
$$Xe_{h,p}X^{-1}=\sum_{i,j}u_{i,h}\bar u_{j,p}e_{i,j},\quad X=\sum_{i,j}u_{i,j} e_{i,j} \in U(d),\ X^{-1}=\sum_{i,j}\bar u_{j,i}e_{i,j} . $$ A basis of representative functions for $R=End(V)$ is
\begin{equation}\label{reU}
tr(e_{i,j} Xe_{h,p}X^{-1})=tr(e_{i,j} \sum_{a,b}u_{a,h}\bar u_{b,p}e_{a,b})=   u_{j,h}\bar u_{i,p},\quad i,j,h,p=1,\ldots,d.
\end{equation}Since a duality between $End(V)^{\otimes k}$ and itself is the non degenerate pairing:$$ \langle A\mid B\rangle:=tr(A\cdot B)$$
 a basis of representative functions of $End(V)^{\otimes k}$ is formed  by the products
$$tr(e_{i_1,j_1}\otimes e_{i_2,j_2}\ldots \otimes e_{i_k,j_k}\cdot Xe_{h_1,p_1}X^{-1}\otimes Xe_{h_2,p_2}X^{-1}\ldots \otimes Xe_{h_k,p_k} X^{-1})=$$
\begin{equation}\label{reUk}
 tr\left(\mathtt e_{\underline i,\underline j} \cdot X \mathtt e_{\underline h,\underline p}X^{-1}\right) =\prod_{\ell=1}^ktr(e_{i_\ell,j_\ell}\cdot Xe_{h_\ell,p_\ell}X^{-1})=\prod_{\ell=1}^ku_{j_\ell,h_\ell}\bar u_{i_\ell,p_\ell},
\end{equation}
where in order to have compact notations we write  \begin{equation}\label{com}
\underline i:=(i_1,i_2,\ldots,i_k),\quad   \mathtt e_{\underline i,\underline j}=e_{i_1,j_1}\otimes e_{i_2,j_2}\ldots \otimes e_{i_k,j_k}.
\end{equation} \begin{equation}\label{com1}
    \mathtt u_{\underline a,\underline b}= \prod_{\ell=1}^ku_{a_\ell,b_\ell}.
\end{equation} Therefore every integral in Formula \eqref{intr}  for $k_1=k_2=k$  is the integral of a  representative function.

Of course the   expression of a representative function as $ tr\left(\mathtt e_{\underline i,\underline j} \cdot X\mathtt e_{\underline h,\underline p} X^{-1}\right)$     is not unique.\smallskip

Collins writes the explicit Formula \eqref{fina2} for $$
\int_{U(d)}\prod_{\ell=1}^ku_{j_\ell,h_\ell}\bar u_{i_\ell,p_\ell} du =\int_{U(d)}\mathtt u_{\underline j,\underline h}\bar{\mathtt u}_{\underline i,\underline p}\,du $$\begin{equation}\label{forint}=\int_{U(d)} tr\left(\mathtt e_{\underline i,\underline j} \cdot X \mathtt e_{\underline h,\underline p}X^{-1}\right)dX =  tr\left(\mathtt e_{\underline i,\underline j} \cdot E(\mathtt e_{\underline h,\underline p})\right)
\end{equation}

 In order to do this, it is enough to have an explicit formula for the equivariant projection $E$ of  $End(V)^{\otimes k}$  to the  $GL(V)$  (or $U(d)$) invariants $\Sigma_k(V)$,  the algebra generated by the permutation operators  $\sigma\in S_k$ acting on $V^{\otimes k}$.\smallskip
 
His idea is to consider first the map
\begin{equation}\label{ilphi}
\Phi:End(V)^{\otimes k}\to\Sigma_k(V),\quad \Phi(A):=\sum_\sigma tr(A\circ \sigma^{-1})\sigma.
\end{equation} This map is a  $GL(V)$ equivariant map  to $\Sigma_k(V),$ but it is not a projection. In fact  restricted to 
$\Sigma_k(V),$ we have
$$\Phi:\Sigma_k(V)\to\Sigma  _k(V),\quad \Phi(\tau):=\sum_{\sigma\in S_k} tr(\tau\circ \sigma^{-1})\sigma.$$ Setting $\quad \sigma=\gamma\tau,\quad \tau \sigma^{-1}=\gamma^{-1}$   we have:
 \begin{equation}\label{mophi}
\Phi(\tau)  =\sum_{\gamma\in S_k}  tr(\gamma^{-1})\gamma\,\tau =\Phi(1)\tau =\tau \Phi(1)=\tau \sum_{\gamma\in S_k}  tr(\gamma^{-1})\gamma.
\end{equation}We have seen, in Corollary \ref{invp}, that $$\Phi(1)=\sum_{\gamma\in S_k}  tr(\gamma^{-1})\gamma= \sum_{\gamma\in S_k}  d^{c(\gamma)}\gamma$$  is a central invertible  element of  $\Sigma_k(V)$.  So the equivariant projection $E$  is $\Phi$ composed with multiplication by the inverse $W\!g(d,k)$  of the element  $\Phi(1)=\sum_{\gamma\in S_k}  tr(\gamma^{-1})\gamma$ given by Formula \eqref{rmene2}  or \eqref{rmene}.
\begin{equation}\label{eqp}
E=(\sum_{\gamma\in S_k}  tr(\gamma^{-1})\gamma)^{-1}\circ \Phi=\Phi(1)^{-1}\circ\Phi=W\!g(d,k)\circ \Phi.
\end{equation}

Of course  
$$\Phi(\mathtt e_{\underline h,\underline p} )= \sum_\sigma tr(\mathtt e_{\underline h,\underline p}\circ \sigma^{-1})\sigma $$
$$\implies E(\mathtt e_{\underline h,\underline p})=\sum_{\gamma\in S_k} W\! g(d,\gamma)\gamma  \sum_\sigma tr(\mathtt e_{\underline h,\underline p}\circ \sigma^{-1})\sigma $$  and 
Formula \eqref{forint}  becomes
\begin{equation}\label{fina}
tr(\mathtt e_{\underline i,\underline j} \circ \sum_{\gamma\in S_k}  W\! g(d,\gamma)\gamma  \sum_\sigma tr(\mathtt e_{\underline h,\underline p}\circ \sigma^{-1})\sigma )
\end{equation}\begin{equation}\label{fina1}
=\sum_{\gamma,\sigma \in S_k}  tr(\mathtt e_{\underline i,\underline j}\circ\gamma  )  tr(\mathtt e_{\underline h,\underline p}\circ \sigma^{-1}) W\! g(d,\gamma\sigma^{-1})
\end{equation}
  
From Formulas \eqref{formuu0} and \eqref{formuu}  since $e_{i,j}=e_i\otimes e^j$  we have
$$tr(e_{i_1,j_1}\otimes e_{i_2,j_2}\ldots \otimes e_{i_k,j_k}\circ \gamma)=\prod_h\langle e_{i_{\gamma(h)} }\mid e^{j_h}\rangle=\prod_h\delta_{i_{\gamma(h)}}^{j_h}$$
$$\eqref{fina1}=\!\! \!\! \sum_{\gamma,\sigma \in S_k}\!\! \prod_\ell\delta_{i_{\gamma(\ell)}}^{j_\ell}\prod_\ell\delta_{h_{\ell}}^{p_\sigma(\ell)} W\! g(d,\gamma\sigma^{-1})$$ \begin{equation}\label{fina2}
\implies \int_{U(d)}\mathtt u_{\underline j,\underline h}\bar{\mathtt u}_{\underline i,\underline p}\,du = \boxed{\sum_{\gamma,\sigma \in S_k}  \delta_{\gamma(\underline i)}^{\underline j}\delta_{\underline h}^{\sigma(\underline  p)} W\! g(d,\gamma\sigma^{-1})}.
\end{equation}
\begin{remark}\label{part}
In particular for  $i_\ell  =h_\ell  =p_\ell  =\ell$  and $j_\ell=\tau(\ell),\ 1\leq\ell\leq k$, Formula  \eqref{fina2} gives $W\! g(d,\tau)$.
\end{remark}

Collins then goes several steps ahead since he is interested in the asymptotic  behaviour of this function as $d\to \infty$ and proves an asymptotic expression  for any $\sigma$ in term of its cycle decomposition, Theorem \ref{tcoll2}.\qed 
\subsection{Tensor polynomials}
In work in  progress with Felix Huber, \cite{HP},   we consider the problem of  understanding $k$--tensor valued polynomials of  $n$, $d\times d$  matrices. 

That is maps   from $n$ tuples of $d\times d$  matrices $x_1,\ldots,x_n\in End(V)$  to tensor space $End(V)^{\otimes k}$ of the form
$$G(x_1,\ldots,x_n)=\sum_i\alpha_i m_{1,i}\otimes m_{2,i}\otimes \ldots \otimes m_{k,i},\ \alpha_i\in\mathbb C\quad m_{j,i}\quad\text{monomials in the} \ x_i.$$

A particularly interesting case is when  the polynomial is  multilinear and alternating in  $n=d^2$ matrix variables. 

In this case, by Proposition \ref{ilmulti} we have
\begin{theorem}\label{dq}
\begin{enumerate}\item $$G(x_1,\ldots,x_{d^2})= \det(x_1,\ldots,x_{d^2})\bar J_G.$$
\item Moreover we have the explicit formula
$$G(e_{1,1},  e_{1,2}, e_{2,1}, e_{2,2}, \ldots, e_{d,d})=\bar J_G.$$

\item The element  $\bar J_G\in M_d^{\otimes k}$ is $GL(k)$ invariant and so  $\bar J_G\in \Sigma_k (V)$ is a linear combinations of the  elements  of the symmetric group $S_n\subset    M_d^{\otimes k}$ given by the permutations.  

 \end{enumerate}
\end{theorem}
For theoretical reasons instead of computing $\bar J_G$ it is better to compute its multiple, as in Formula \eqref{costdis}:
\begin{equation}\label{verJ}
G(x_1,\ldots,x_{d^2})= \mathcal T_d
(X) J_G,\quad \bar J_G=\mathcal C_d J_G.\end{equation}
 Using Formula  \eqref{ilphi} we may first compute
 $$\Phi(G(x_1,\ldots,x_{d^2}))=\sum_{\sigma\in S_k}tr(\sigma^{-1}\circ  G(x_1,\ldots,x_{d^2}))= \mathcal T_d
(X) \Phi(J_G).$$

Consider the special case   \begin{equation}\label{ilgz}
G_d(Y_1,\ldots,Y_{d^2}):=  Alt_Y(m_1(Y)\otimes \dots\otimes  m_d(Y)),\quad                   m_i(Y)=Y_{(i-1)^2+1}\ldots Y_{i^2}. 
\end{equation}
\begin{lemma}\label{immo}   
 
\begin{equation}\label{eese}Alt_Ytr(\sigma^{-1}\circ m_1(Y)\otimes \dots\otimes  m_d(Y)) =\begin{cases} \mathcal T_d(Y)\quad \text{if  }\  \ \ \sigma=1\\0\quad \text{otherwise}\end{cases} \end{equation} 
  \end{lemma}
\begin{proof}
$$ tr(\sigma^{-1}\circ m_1(Y)\otimes \dots\otimes  m_{d}(Y)) =\prod_{i=1}^jtr(N_i)$$
with $N_i$  the product of the monomials $m_j$ for $j$  in the $i^{th}$ cycle of $\sigma$, cf. Formula \eqref{phis}.  The previous invariant   gives by alternation the invariant
$$Alt_Y\prod_{i=1}^jtr(N_i)=  T_{a_1}\wedge T_{a_2}\wedge\cdots\wedge T_{a_j},\quad a_i=\text{degree of}\ N_i$$  in degree $d^2$.  If    $\sigma\neq 1$ we have $j<d$   hence the product is 0, since  the  only invariant alternating in  this degree is  $T_1\wedge T_3\wedge T_5\wedge \ldots \wedge T_{2d-1}$. 

On the other hand if  $\sigma =1$   we have $N_i=m_i$  and the claim follows. \end{proof}
\begin{proposition}\label{forgz}
We have 
\begin{equation}\label{forgz1}
G_d(Y_1,\ldots,Y_{d^2}):=  Alt_Y(m_1(Y)\otimes \dots\otimes  m_d(Y))= \mathcal T_d(Y) W\!g(d,d).
\end{equation}
\end{proposition}
\begin{proof}
The previous Lemma in fact implies that $\Phi(G_d(Y_1,\ldots,Y_{d^2}))=\mathcal T_d(Y) 1_d$ therefore $\Phi(J_{G_d})\stackrel{\eqref{mophi}}=\Phi(1)J_{G_d}=1$ so that $J_{G_d}=\Phi(1)^{-1}=W\!g(d,d)$.
\end{proof}

 \subsection{The construction of Formanek\label{sfor}}

   Let us now discuss  a theorem of Formanek relative to a conjecture of Regev, see  \cite{formanek2} or  \cite{agpr}. This states that, a certain explicit central polynomial $F(X,Y)$ in  $d^2$, $d\times d$ matrix variables $X=\{X_1,\ldots,X_{d^2}\}$ and  another  $d^2$,  $d\times d$ matrix variables $Y=\{Y_1,\ldots,Y_{d^2}\}$, is non zero.   This polynomial plays an important role  in the theory of polynomial identities, see \cite{agpr}.

   The definition of $F(X,Y)$ is this, decompose $d^2=1+3+5+\ldots +(2d-1)$ and accordingly decompose the $d^2$ variables $X $ and  the $d^2$ variables $Y$ in the two lists.  Construct the monomials  $m_i(X), i=1,\ldots,d $    and similarly  $m_i(Y) $ as product in the given order of the given   $2i-1$ variables $X_i$ of the $i^{th}$ list as for instance
$$m_1(X)= X_1, m_2(X)= X_2X_3X_4 ,  m_3(X)= X_5X_6X_7X_8X_9 ,\ldots .$$$$m_i(X)=X_{(i-1)^2+1}\ldots X_{i^2},\quad m_i(Y)=Y_{(i-1)^2+1}\ldots Y_{i^2}.$$We finally define\begin{equation}\label{RFa}
F(X,Y):= Alt_XAlt_Y(m_1(X)m_1(Y)m_2(X)m_2(Y)\ldots m_d(X)m_d(Y)),
\end{equation}  where $Alt_X$ (resp. $Alt_Y$) is the operator of alternation, Formula \eqref{alta}, in the variables $X$ (resp. $Y$). By Theorem \ref{dq} it takes scalar values, a multiple of $\mathcal T_d(X)\mathcal T_d(Y) $,  but it could be identically 0.
\begin{theorem}\label{teF}
\begin{equation}\label{RFaFF}
F(X,Y) = (-1)^{d-1} \frac  {1}{(d!)^2 (2d-1)}\mathcal T_d(X)\mathcal T_d(Y) Id_d
\end{equation} 
$$\stackrel{\eqref{costdis}}= (-1)^{d-1} \frac  {\mathcal C_d^2}{(d!)^2 (2d-1)}\Delta(X)\Delta(Y) Id_d;\quad \Delta(X)=\det(X_1,\ldots,X_{d^2}).$$
\end{theorem} Notice that by Formula \eqref{costdis} the coefficient is an integer (as predicted).

Thus $F(X,Y) $ is a central polynomial. In fact it has also the property of being in the {\em conductor } of the ring of polynomials in generic matrices inside the trace ring. In other words  by multiplying  $F(X,Y) $ by any invariant we still can write this as a non commutative polynomial. This follows  by polarizing in $z$  the identity, cf. \cite{agpr} Proposition 10.4.9 page 286.
$$\det(z)^dF(X,Y) =F(zX,Y) =F(X,zY)  =F( Xz,Y) =F(X,  Yz) .$$
\smallskip

Let us follow  Formanek's proof.  First, since $F(x,y)$  is a central polynomial  Formula  \eqref{RFaFF} is equivalent to:\begin{equation}\label{RFaFF1}
tr(F(X,Y) )= (-1)^{d-1} \frac  {d}{(d!)^2 (2d-1)}\mathcal T_d(X)\mathcal T_d(Y) .
\end{equation}
Now we have, with $\sigma_0=(1,2\ldots,d)$ the cycle:
\begin{equation}\label{trdd}
tr(F(X,Y) )=tr(\sigma_0^{-1}\circ Alt_XAlt_Y(m_1(X)m_1(Y)\otimes m_2(X)m_2(Y)\otimes \ldots \otimes m_d(X)m_d(Y)),
\end{equation}
$$\stackrel{\eqref{forgz1}} =tr(\sigma_0^{-1}\circ Alt_X (m_1(X)\otimes m_2(X)\otimes \ldots \otimes m_d(X)\cdot W\!g(d,d))\mathcal T_d(Y)  . $$
Denote  $W\!g(d,d)=\sum_{\tau\in S_d}a_\tau \tau$, we have
$$tr(\sigma_0^{-1}\circ Alt_X (m_1(X)\otimes m_2(X)\otimes \ldots \otimes m_d(X)\cdot W\!g(d,d))$$$$=\sum_\tau a_\tau tr(\sigma_0^{-1}\tau\circ Alt_X (m_1(X)\otimes m_2(X)\otimes \ldots \otimes m_d(X))$$ which, by Lemma \ref{immo}  equals $a_{\sigma_0}\mathcal T_d(X) . $
Therefore the main Formula \eqref{RFaFF} follows from  Formula \eqref{nfaaz}.

 \section{Appendix}
If $k>d$ of course there is still an expression as in   Formula \eqref{wg} but it is not unique.

It can be made unique by a choice of a basis of  $\Sigma_{k}(V)$.  This may be done as follows.
\begin{definition}
Let $0<d$ be an integer and let $\sigma\in S_n$. 

Then $\sigma$ is called
{\em $d$--bad}\index{$d$--bad} if $\sigma$ has a descending subsequence of length $d$,
namely, if  there exists a sequence     $1\le i_1<i_2<\cdots <i_d\le n$ such that
$\sigma(i_1)>\sigma(i_2)>\cdots
>\sigma(i_d)$. Otherwise $\sigma$ is called {\em $d$--good}.

\end{definition}
\begin{remark}
$\sigma$ is $d$--good if any descending sub--sequence of $\sigma$ is of
length $\le d-1$. If $\sigma$ is $d$-good then $\sigma$ is $d'$-good for any
$d'\ge d$.

Every permutation is $1$-bad.
\end{remark}
\begin{theorem}\label{dgoo}If $\dim(V)=d$ the $d+1$--good permutations form a basis of $\Sigma_{k}(V)$.

\end{theorem}
\begin{proof}
Let us first prove that the $d+1$--good permutations span $\Sigma_{k,d}$. 

So let 
 $\sigma$ be $d+1$--bad so that there exist  $1\le i_1<i_2<\cdots <i_{d+1}\le n$ such that
$\sigma(i_1)>\sigma(i_2)>\cdots
>\sigma(i_d+1)$.   If $A$  is the antisymmetrizer  on the $d+1$ elements $\sigma(i_1),\sigma(i_2),\cdots
,\sigma(i_d+1)$ we have that $A\sigma=0$ in $\Sigma_{k}(V)$, that is, in $\Sigma_{k}(V)$, $\sigma$ is a linear combination of permutations obtained from the permutation $\sigma$ with some proper rearrangement of the indices  $\sigma(i_1),\sigma(i_2),\cdots
,\sigma(i_d+1)$.These permutations are all lexicographically $<\sigma$.  One applies the same algorithm  to any of these permutations which is still  $d+1$--bad. This gives an explicit algorithm which stops when  $\sigma$ is expressed as a linear combination of   $d+1$--good permutations (with integer coefficients so that the algorithm works in all characteristics).

In order to prove that the  $d+1$--good permutations form a basis, it is enough to show that their number equals the dimension of $\Sigma_{k,d}$. This is insured by a classical result of Schensted which we now recall.
\end{proof}
\subsubsection{The RSK  and $d$-good permutations
\label{Robinson--Schensted}}
The RSK correspondence\footnote{Robinson, Schensted, Knuth}, see~\cite{knuth1},~\cite{stanley},
 is a combinatorially defined bijection
$\sigma\longleftrightarrow (P_\lambda,Q_\lambda)$ between permutations $\sigma\in S_n$ and pairs
$P_\lambda,Q_\lambda$ of standard tableaux of same shape $\lambda,$ where $\lambda\vdash n$.

In fact more generally it associates to a word, in the free monoid, a pair of tableaux, one standard and the other semistandard filled with the letters of the word. This correspondence may be viewed as a combinatorial counterpart to the Schur--Weyl and Young theory.

The correspondence is  based on a simple {\em game} of inserting a letter.  

We  have some  letters piled up so that lower letters appear below higher letters and we want to insert a new letter $x$.  If $x$ fits on top of the pile we place it there otherwise we go down the pile, until we find a first place where we can replace the existing letter with $x$.  We do this and expel that letter, first creating a new pile or, if we have a second pile of letters then we try to place that letter there and so on.

So let us pile inductively the word  $strange$.
  $$e\mapsto e,\ g\mapsto \begin{matrix}
g\\e
\end{matrix} ,\ n\mapsto \begin{matrix}n\\
g\\e
\end{matrix},\ a\mapsto \begin{matrix}
n\\g\\a&e
\end{matrix},\ r\mapsto \begin{matrix}
r\\n\\g\\a&e
\end{matrix}  ,\ t\mapsto \begin{matrix}
t\\r\\n\\g\\a&e
\end{matrix} ,\ s\mapsto \begin{matrix}
s\\r\\n\\g&t\\a&e
\end{matrix} .   $$  Notice that, as we proceed, we can keep track of where we have placed the new letter, we do this by filling a corresponding tableau.
$$ \begin{matrix}
6\\5\\3\\2&7\\1&4
\end{matrix},\quad \begin{matrix}
s\\r\\n\\g&t\\a&e
\end{matrix} .  $$ It is not hard to see that from the two tableaux one can {\em decrypt}  the word we started from giving the bijective correspondence.\smallskip

Assume now that $\sigma\longleftrightarrow (P_\lambda,Q_\lambda)$, where $P_\lambda,Q_\lambda$ are standard tableaux, given by the RSK correspondence.
By a classical theorem of Schensted~\cite{schensted}, $ht(\lambda)$ equals the
length of a longest decreasing subsequence in the permutation $\sigma$.
Hence $\sigma$ is $d+1$-good if and only if $ht(\lambda)\le d$.
\smallskip

Now  $M_\lambda$ has a basis indexed by  standard    tableaux of shape $\lambda$, see \cite{P7}. Thus the algebra $\Sigma_k(V)$  has a basis indexed by   pairs of tableaux of shape $\lambda.\ ht(\lambda)\leq d$ and the claim follows.\qed\smallskip

Therefore one may define the Weingarten function for all $k$ as a function on the $d+1$--good permutations in $S_k$.

\subsubsection{Cayley's $\Omega$ process}It may be interesting to compare the method  of computing  the integrals of Formula \eqref{eqp} with a very classical approach used by the $19^{th}$  century invariant theorists.

Let me recall this for the modern  readers.
Recall first that, given a $d\times d$  matrix   $X=(x_{i,j})$,  its {\em adjugate} is $\bigwedge^{d-1}(X)=(y_{i,j})$  with $y_{i,j}$ the {\em cofactor} of $x_{j,i}$ that is $(-1)^{i+j}$  times the determinant of the minor of $X$ obtained by removing the $j$ row and $i$ column.   Then the inverse of $X$ equals  $\det(X)^{-1}\bigwedge^{d-1}(X)$.

It is then easy to   see that, substituting to $u_{i,j}$ the variables $x_{i,j}$  and to  $\bar u_{i,j}$ the polynomial $y_{i,j}$ one transforms a monomial $M=\prod_{\ell=1}^ku_{j_\ell,h_\ell}\bar u_{i_\ell,p_\ell}$   into a  polynomial $\pi_d(M)$ in
 the variables    $x_{i,j}$   homogeneous of degree $dk$, the invariants under $U_d$  become  powers  $\det(X)^k$.  Denote by $S^{kd}(x_{i,j})$ the space of these polynomials  which,   under the action of $GL(d)\times GL(d)$,  decomposes  by Cauchy formula, cf. Formula 6.18, page 178,  of \cite{agpr}. Then we have also an equivariant projection from these polynomials to the 1--dimensional space  spanned by $\det(X)^k$, it is given through the Cayley $\Omega$ process used by Hilbert in his famous work on invariant theory.  The $\Omega$ process is the differential operator given    by the determinant of the matrix of derivatives:
 \begin{equation}\label{ome}X=(x_{i,j}),\quad Y= (\pd{}{x_{i,j}}),\quad
\Omega:=\det(Y).
\end{equation}
We  have that  $\Omega^k$ is  equivariant under the action by $SL(n)$ so it maps to 0 all the irreducible representations different from the 1--dimensional space  spanned by $\det(X)^k$ while
$$\Omega\det(X)^k=k(k+1)\ldots (k+d-1)\det(X)^{k-1}. $$ 
Both statements follow from the Capelli identity, see \cite{P7} \S  4.1 and  \cite{cape}.
$$\boxed{\det(X)\Omega=\det(a_{i,j})},\ a_{i,i}=\Delta_{i,i}+n-i,\ a_{i,j}=\Delta_{i,j},\ i\neq j $$
$$\text{the polarizations}\quad \Delta_{i,j}=\sum_{h=1}^d x_{i,h}\pd{}{x_{h,j}}.  $$ If we denote by $\underline x_i:=(x_{i,1},\ldots,x_{i,n})$  we have the Taylor series for a function $f(\underline x_1,\ldots,\underline x_n)$ of the vector coordinates $\underline x_i$. 
$$ f(\underline x_1,\ldots,\underline x_j+\lambda \underline x_i,\ldots,\underline x_n)=\sum_{k=0}^\infty\frac{(\lambda  \Delta_{i,j})^k}{k!} f(\underline x_1,\ldots,\underline x_n).$$Thus
\begin{equation}\label{ome1}
\int_UM\,du=\frac  {\Omega^k\pi_d(M)}{\prod_{i=1}^k(i(i+1)\ldots (i+d-1))} .
\end{equation}
We can use Remark \ref{part} to  give a possibly useful formula:
\begin{equation}\label{wein}
W\! g(d,\gamma)=\frac {\Omega^k\pi_d(M)}{\prod_{i=1}^k(i(i+1)\ldots (i+d-1))},\ M=\prod_{i=1}^ku_{i,i}\bar u_{i,\gamma(i)}.
\end{equation} 
 \smallskip
 
  Let me discuss a bit  some calculus with these operators.
 
\begin{lemma}\label{ll}
If $i\neq j $ then $\Delta_{ij}$ commutes with 
 $\Omega $ and with $ \det(X) $   while  \begin{equation}\label{ccc}
[ \Delta_{ii},\det(X)]=\det(X),\quad [\Delta_{ii} ,\Omega]=- \Omega  .
\end{equation}

\end{lemma} 
\begin{proof}
 The operator $\Delta_{ij}$ commutes with all of the columns
of $\Omega$ except the $i^{th}$ column $\omega_i$ with entries $\pd{}{x_{it}}$. Now 
$[\Delta_{ij},\pd{}{x_{it}}]=-\pd{}{x_{jt}}$, from which $[\Delta_{ij},\omega_i]=-\omega_j.$ The
result follows immediately.
\end{proof}  
 Let us introduce a more general determinant, analogous to a
characteristic polynomial. We denote it by $C_m(\rho)=C (\rho)$ 
and define it as:
$$\begin{pmatrix}
 \quad\ \ \ \Delta_{1,1}+m-1+\rho &\Delta_{1,2}\phantom{_{-1}}
&\dots& \Delta_{1,m}\phantom{_{-1}}\\
\Delta_{2,1}\phantom{-1}&\quad\ \ \ \Delta_{2,2}+m-2+\rho
&\dots& \Delta_{2,m}\phantom{_{-1}}\\
\dots&\dots&\dots& \dots \\
\dots&\dots&\dots& \dots \\
\Delta_{m-1,1}&\Delta_{m-1,2}&\dots&  \Delta_{m-1,m}
\\
\Delta_{m,1} \phantom{_{-1}}&\Delta_{m,2}\phantom{_{-1}} 
&\dots &\Delta_{m,m}+\rho\phantom{_{-1}}
\end{pmatrix}.$$  We have now a generalization of the
Capelli identity:
 \begin{proposition}\label{capell}
$$  \Omega C(k) =C(k+1)\Omega ,\qquad   \det(X) C(k) =C(k-1) \det(X)  $$ 
 $$  \det(X) ^k\Omega^k=C(-(k-1))C(-(k-2))\dots C(-1)C,$$$$
\Omega^k \det(X) ^k=C(k)C(k-1)\dots C(1).$$
\end{proposition}

\begin{proof}
We may apply directly   Formulas \eqref{ccc} and then proceed by induction.   

\end{proof}
 
Develop now $C_m(\rho)$ as a polynomial in $\rho$ obtaining an expression
 
$$C_m(\rho)=\rho^m+\sum_{i=1}^m K_i\rho^{m-i}.$$
Capelli proved, \cite{cape}, that, as the elementary symmetric functions generate the algebra of symmetric functions so the elements $K_i$  generate the center of the enveloping algebra of the Lie algebra of matrices.

In \cite{P7} Chapter 3, \S 5  it is also given the explicit formula, also due to Capelli,  of the action of $C_m(\rho)$ (as a scalar) on the irreducible representations  which classically appear as {\em primary covariants}.
 \subsection{A quick look at the symmetric group}
\subsubsection{The branching rule and Young basis}
Recall that the irreducible representations of $S_n$  over $\mathbb Q$  are indexed by partitions of $n$ usually displayed as {\em Young diagrams}.  

The Branching rules, see   \cite{sagan},  \cite{macdonald} or  \cite{P7},   tell us how the representation  $M_\lambda$ decomposes once we   restrict  to $S_{k-1}$.  The  irreducible representation $M_\lambda$  becomes the direct sum $\oplus_{\mu\subset\lambda,\ \mu\vdash k-1}M_\mu.$  The various $\mu$ are obtained from $\lambda$  by marking one corner box with $k$  and removing this box.
\bigskip

\smallskip

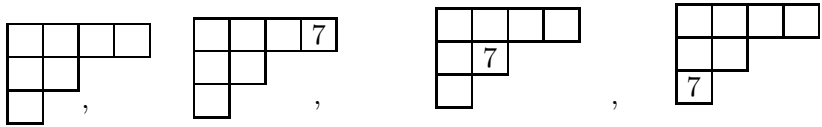
\vbox{  \begin{Young}
  &&&\cr
  &\cr
  \cr
\end{Young}\vskip-1.45cm \hskip1.5cm\begin{Young} 
  && &7 \cr
  &\cr
  \cr
\end{Young}\vskip-1.50cm  \quad  ,\quad\quad\quad\quad\quad \quad \quad\quad\quad\quad  \quad \begin{Young} 
  && & \cr
  &7\cr
  \cr
\end{Young}\vskip-1.50cm   \quad\quad \quad  \quad\quad\quad \quad\quad\quad ,\quad\quad\quad\quad\quad \quad \quad\quad\quad\quad ,\quad \begin{Young} 
  && & \cr
  & \cr
 7 \cr
\end{Young}}\bigskip

$$M_{4,2,1}=M_{3,2,1}\oplus M_{4,1,1}\oplus M_{4,2 }$$
This can be repeated on each summand decomposed into irreducible representations of  $S_{n-2}$

\bigskip

\smallskip

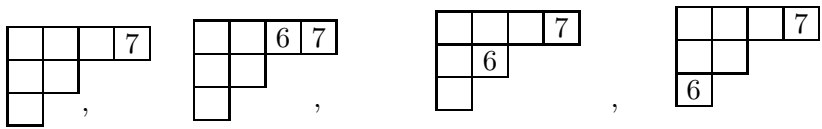
\vbox{  \begin{Young}
  &&&7\cr
  &\cr
  \cr
\end{Young}\vskip-1.45cm \hskip1.5cm\begin{Young} 
  && 6&7 \cr
  &\cr
  \cr
\end{Young}\vskip-1.50cm  \quad  ,\quad\quad\quad\quad\quad \quad \quad\quad\quad\quad  \quad \begin{Young}   && &7 \cr
  &6\cr
  \cr
\end{Young}\vskip-1.50cm   \quad\quad \quad  \quad\quad\quad \quad\quad\quad ,\quad\quad\quad\quad\quad \quad \quad\quad\quad\quad ,\quad \begin{Young} 
  && &7 \cr
  & \cr
 6 \cr
\end{Young}}\bigskip

$$M_{3,2,1}=M_{2,2,1}\oplus M_{3,1,1}\oplus M_{3,2 }$$
After $k-1$ steps  we have a list of {\em skew  standard tableaux}  filled with the numbers $n,n-1,\ldots,n-k+1$ so that removing the boxes occupied by these numbers we still have a Young diagram  and these tableaux index  a combinatorially defined decomposition of $M_\lambda$  into irreducinle representations of $S_{n-k }$.
Getting, after $n$ steps  a decomposition of $M_\lambda$ into one dimensional subspaces indexed by {\em standard tableaux},   as out of a total of 35:
\bigskip

\smallskip

\vbox{  \begin{Young}
  1&4&6&7\cr
  2&5\cr
  3\cr
\end{Young}\vskip-1.45cm \hskip1.5cm\begin{Young} 
  1&3& 6&7 \cr
  2&5\cr
 4 \cr
\end{Young}\vskip-1.50cm  \quad  ,\quad\quad\quad\quad\quad \quad \quad\quad\quad\quad  \quad \begin{Young} 
 1 &2&5 &7 \cr
  3&6\cr
  4\cr
\end{Young}\vskip-1.50cm   \quad\quad \quad  \quad\quad\quad \quad\quad\quad ,\quad\quad\quad\quad\quad \quad \quad\quad\quad\quad ,\quad \begin{Young}   1&3& 5&7 \cr
  2&4 \cr
 6 \cr
\end{Young}}\bigskip

\begin{equation}\label{gz}
M_\lambda=\oplus_{T\in \ \text{standard tableaux} }M_T,\ \dim_{\mathbb Q}M_T=1.
\end{equation}
In fact there is a scalar product  on $M_\lambda$  invariant under $S_n$ and unique up to scale for this property. The decomposition is then into  orthogonal one dimensional subspaces. One then may choose a basis element  $v_T$ for the one dimensional subspace indexed by $T$ with $|v_T|=1$  but allowing to work on some real algebraic extension of $\mathbb Q$. This is then unique up to sign.  

\begin{remark}\label{td}
Observe that, given a standard tableau $T$ and a number $k\leq n$  the space $M_T$ lies in the irreducible representation of  $S_k$  associated to the skew tableau obtained form $T$ by emptying all the boxes with the numbers $i\leq k$. Its Young diagram is  the diagram containing the indices from $1,\ldots,k$ in $T$. As example the first tableau of the previous list lies in an  irreducible representation of $S_5$  of partition $2,2,1$ and one  of $S_4$  of partition $2,1,1$;  while the third  $3,1,1$ and again  $2,1,1$ but different from the previous one since they are associated to different skew tableaux.
\end{remark}

\subsubsection{A maximal commutative subalgebra}
Denote by $\mathcal Z_n$  the center of the group algebra $\mathbb Z[S_n]$ it is the free abelian group with basis the class functions. 
A basic Theorem of Higman and Farahat \cite{fh}, states that the elements $C_j$  generate (over $\mathbb Z$) as algebra the center $\mathcal Z_n$  of  $\mathbb Z[S_n]$.  

Now consider the inclusions  $S_1\subset S_2\subset\ldots\subset S_{n-1} \subset S_n$  which induces inclusions  
$\mathcal Z_j\subset \mathbb Z[S_n], \ j=1,\ldots,n$. 
\begin{definition}\label{matk} We define   $\mathfrak Z_n$ to be the (commutative)  algebra generated by all the algebras  $\mathcal Z_j$.

\end{definition}

\begin{corollary}\label{maa0}  

The 1--dimensional subspaces $M_T$ associated to standard tableaux are eigenspaces for $\mathfrak Z_n$.
\end{corollary}
\begin{proof}
Take one such 1--dimensional subspace  $M_T$ associated to a standard tableau $T$. Given any $k\leq n$  the space $M_T$ by construction is contained in an irreducible representation of $S_k$  where the elements of $ \mathcal Z_k$ act as scalars.

 \end{proof}
By the Theorem of Jucys--Murphy  and the Theorem of Farahat--Higman the subalgebra of $\mathbb Z[S_n]$  generated by the elements $J_2,\ldots, J_k$      contains the   class algebra   $\mathcal Z_k$ (and conversely). in the next Theorem \ref{JMu} we will see that in fact this subalgebra is maximal semisimple.

The final analysis is to understand the eigenvalues of the operators $J_i$ which generate $\mathfrak Z_n$ on  $M_T$.
Given a standard Tableau $T$  and a number $i\leq n$  this number appears in one specific box of the diagram of $T$  and then we define $c_T(i)$  to be the content of this box as in Formula  \eqref{eqn7}.

As example for the first tableau  of the list before Formula \eqref{gz}
$$ c_T(1)=0,\ c_T(2)=-1,\ c_T(3)=-2,\ c_T(4)=1,\ c_T(5)=0,\ c_T(6)=2,\  c_T(7)=3.$$

Let us start with the following fact.  Denote by $c_2(k)$ the sum of all transpositions of $S_k$. It is a central element so it acts as a scalar on each irreducible representation and one has, see Frobenius \cite{frobenius} or Macdonald  \cite{macdonald}
    \begin{proposition}\label{tra}
The action of $c_2(k)$ on an irreducible representation associated to a partition $\lambda=\lambda_1,\ldots,\lambda_k$ is
\begin{equation}\label{fotr}
\frac12\sum_{i=1}^k(\lambda_i^2-(2i-1)\lambda_i)
\end{equation}

\end{proposition}
 
If  we consider $S_{k-1}\subset  S_k $  we have  $J_k= c_2(k)-c_2(k-1) .$

\begin{theorem}\label{JMu}
\begin{equation}\label{JMe}
J_iv_T= c_T(i)v_T,\ \forall i=2,\ldots,n,\ \forall T.
\end{equation}
\end{theorem}
\begin{proof} We follow Okounkov  \cite{Ok}  who makes reference to Olshanski \cite{Ol}.

We need to compute $ (c_2(i)-c_2(i-1) )v_T.$
Now  $v_T$  belongs to the irreducible representation of $S_i$ whose diagram is the subdiagram $D_i$ of the diagram of $T$   containing the indices $1,\ldots, i$  and let  $(a,b)$  be the coordinates of the box  where $i$ is placed.

In the same way  $v_T$  belongs to the irreducible representation of $S_{i-1}$  whose diagram is the subdiagram of $D_i$ obtained removing the box  $(a,b)$.

Applying Formula \eqref{fotr} to the two elements  $  c_2(i),c_2(i-1) $ we see that the two diagrams coincide except for the $a$ row which in one case has length $b$ in the other $b-1$  so the difference of the two values is
$$\frac12[ (b^2-(2a-1)b)-((b-1)^2-(2 a-1(b-1))]=b-a.
 $$
\end{proof}
\begin{proposition}\label{det}
The function $c_T(i),\ i=1,\ldots, n$ determines the standard tableau $T$.
\end{proposition}
\begin{proof}
By induction   the function $c_T(i),\ i=1,\ldots, n-1$ determines the part $T'$ of the tableau $T$  except the box occupied by $n$.

As for this box  we know its content, $c_T(n)$. Now the boxes with a given content form a {\em diagonal} and then the box for $T$  must be the first in this diagonal which is not in  $T'$.
\end{proof}
This shows that the algebra generated by the elements $J_i$  separates all the vectors of all Young bases so:  \begin{corollary}\label{maa} The   elements $J_i,\ i=2, \ldots$.  generate  the  maximal semisimple commutative subalgebra $\mathcal S$ of $\mathbb Q[S_n]$  of all elements which are diagonal  on all Young bases..

\end{corollary}
 \begin{proof}
By Theorem \ref{JMu} and Proposition \ref{det} the subalgebra   $\mathcal S$  maps surjectively to the subalgebra of  $\mathbb Q[S_n]$  of all elements which are diagonal  on all Young bases. But this map is also injective since an element of $\mathbb Q[S_n]$ which vanishes on all irreducible representations equals to 0.  Hence $\mathcal S$ is the direct sum of the diagonal matrices  (in this basis) for all matrix algebras in which   $\mathbb Q[S_n]$ decomposes and this is a maximal commutative semisimple subalgebra  hence the claim.
\end{proof}\subsection{Stanley hook--content formula\label{shcf}}
Let us finally show that the Jucys factorization, Formula \eqref{JF},  can be viewed as a refinement of Stanley hook--content formula \eqref{sth}.

In fact  consider the scalar value of the central operator  $$P=\sum_{\rho\in S_k} d^{c(\rho|}\rho=d\prod_{i=2}^k(d+J_i)$$ on an irreducible representation $M_\mu$. It can be evaluated, from Formula \eqref{dcro1} as
$$\chi_\mu(1)^{-1} tr(P)=\chi_\mu(1)^{-1} \sum_\sigma\sum_{\lambda\vdash k,\ ht(\lambda)\leq d}s_{\lambda }(d)\chi_\lambda(\sigma)
\chi_\mu(\sigma)$$\begin{equation}\label{posc}
=\chi_\mu(1)^{-1} k!s_{\mu }(d)=\prod_{u\in \mu}h_us_{\mu }(d).
\end{equation} On the other hand this scalar is also the value obtained by applying  the operator  $P= d\prod_{i=2}^k(d+J_i)$ on any standard tableau of the Young basis of $M_\mu$ giving, by Formula \eqref{JMe},  the value 
 \begin{equation}\label{cco}
d\prod_{i=2}^k(d+c_T(i))=\prod_{u\in \mu} (d+c_u).
\end{equation}Comparing Formulas  \eqref{posc} and \eqref{cco} one finally has    Stanley hook--content formula \eqref{sth}.

\bibliographystyle{amsalpha}

\end{document}